\documentclass[11pt]{article}

\usepackage{geometry}
\usepackage[dvipdfmx,hiresbb]{graphicx}
\usepackage{amsmath,amssymb,amsthm,mathtools,mathrsfs,bm}
\usepackage{booktabs,microtype}
\usepackage[font=small]{caption}
\usepackage{enumitem}
\setlist[itemize]{itemsep=2pt, topsep=4pt}
\usepackage[hang,flushmargin]{footmisc}
\usepackage{ascmac,fancybox,framed,color}
\usepackage{comment,lastpage,fancyhdr}
\usepackage{natbib}
\usepackage[colorlinks=true,linkcolor=blue,citecolor=blue,urlcolor=blue]{hyperref}

\theoremstyle{plain}
\newtheorem{thm}{Theorem}[section]
\newtheorem{cor}[thm]{Corollary}
\newtheorem{lem}[thm]{Lemma}

\theoremstyle{definition}
\newtheorem{ass}{Assumption}
\theoremstyle{remark}
\newtheorem{rem}{Remark}

\newcommand{\R}{\mathbb{R}}
\newcommand{\N}{\mathbb{N}}

\newcommand{\rk}{\mathrm{rank}}
\newcommand{\tr}{\mathrm{Tr}}
\newcommand{\diag}{\mathrm{diag}}

\newcommand{\E}{\mathrm{E}}
\newcommand{\Var}{\mathrm{Var}}

\newcommand{\Hc}{\mathcal{H}}

\title{Hybrid principal component analysis in\\ multivariate allometric regression}
\author{Koji Tsukuda\thanks{Faculty of Mathematics, Kyushu University, Fukuoka, Japan.} \and Shun Matsuura\thanks{Faculty of Science and Technology, Keio University, Kanagawa, Japan}}

\begin{document}
\maketitle

\begin{abstract}
\noindent
In biological data from allometry studies, the largest eigenvalue is typically dominant, and the gaps between minor eigenvalues are often narrow. 
Such proximity among small minor eigenvalues can lead to instability in statistics based on their corresponding eigenvectors.
This study derives the asymptotic normality of the hybrid principal component analysis estimator of the leading principal eigenvector in the multivariate allometric regression model and proposes a test based on a geometric statistic for the parallelism between the regression direction and the principal component direction  that avoids this instability.
Using the hybrid principal component analysis framework, we analyze the well-known painted turtle carapace data and confirm previously reported results on the allometric extension relationship between female and male turtles.
\end{abstract}

\noindent
keywords: allometric extension; reduced-rank regression; weak identifiability; weighted chi-squared test

\newpage
\section{Introduction}\label{sec:1}
Data obtained from biological measurements often exhibit specific structures that reflect underlying biological mechanisms.
In particular, when analyzing size and shape, it is natural to assume that the direction of growth coincides with the direction of dominant variation, and this type of alignment is commonly assumed in the field of allometry.
This alignment is referred to as ``allometric extension'' \citep{RefBFN99,RefH06}.
Using the well-known painted turtle carapace data \citep{RefJM60}, \citet{RefBFN99} examined the relationship between the two subgroups of male and female turtles. 
They found that while the relationship of allometric extension held for a subset of two variables, it no longer held when all three variables were incorporated into the analysis.
In biological applications, however, the minor eigenvalues of the covariance matrix are often close to each other.
From a statistical perspective, such proximity is known to degrade the performance of inference procedures; see, for example, \citet{RefS03}.
This statistical challenge, together with the commonly held biological expectation that shape variations do not obscure size variation \citep[cf.][]{RefS86}, raises questions about the statistical conclusion of \citet{RefBFN99}.
Motivated by this, we develop a test within the multivariate reduced-rank regression framework and apply it to reexamine the allometric extension in the turtle carapace data.

In morphometrics and allometry studies, while two-dimensional data have long been the focus, extensions to higher-dimensional settings and statistical inferences on the direction of principal components have been developed through several approaches \citep{RefJ63,RefJ84,RefK96}.
Within multivariate allometry, it is commonly observed that the first principal component accounts for a substantial, often dominant, proportion of the total variation and is interpreted as a size-related direction; see, for example, \citet{RefS86,RefK16}.
The size-related direction represented by the first principal component is often associated with observed explanatory variables such as age, sex, and other biological covariates, providing a natural link between principal component analysis (PCA for short) and regression modeling.
Specifically, \citet{RefK22} systematically compares several methods, including PCA and multivariate linear regression, that provide different operational definitions of allometric directions in geometric morphometrics, and demonstrates that these directions do not necessarily coincide in the presence of residual variation.
Despite being widespread in practice, the assumption that the regression direction aligns with the leading principal component has rarely been subjected to formal statistical testing in regression settings. 
From a statistical viewpoint, this assumption is better regarded as a structural hypothesis rather than an a priori truth.
This perspective motivates the development of formal inferential procedures to assess this hypothesis, particularly in regression settings.

The allometric extension model, which assumes the coincidence of the direction of mean differences and the direction of maximal variation, is formulated by restricting the differences in population mean vectors to a one-dimensional space parallel to the leading principal eigenvectors of the population covariance matrices~\citep{RefF97,RefBFN99,RefH06}.
This model is positioned as a relaxation of the common principal components restriction on the one-dimensional CPC subspace model~\citep{RefFNP95}.
Based on this formulation, the allometric extension model and its regression variant proposed by \citet{RefTI06}, the multivariate allometric regression model, have been actively studied \citep{RefBFN99,RefKGT25,RefKHF08,RefMK14,RefS17,RefTM23,RefTM25}.
A distinctive feature of these models is that PCA incorporates mean differences, enabling the construction of a hybrid estimator.
Among these studies, for the two-sample problem, \citet{RefBFN99} utilized the asymptotic normality of the principal eigenvectors established by \cite{RefA63} to propose a Wald-type statistic for testing the coincidence of the leading eigenvectors of the two populations with the mean difference vector.
However, the corresponding testing problem has not been addressed in the context of multivariate allometric regression. 
Although \citet{RefTI06} proposed a bootstrap-based method and \citet{RefS17} developed likelihood ratio and score tests under a random explanatory variables setting, their methods are not based on the differences between the regression directions and the principal eigenvectors. 
Consequently, their approaches rely on a fundamentally different perspective from that of \citet{RefBFN99}.

To fill this gap, within the reduced-rank regression framework of rank one, this study proposes a direct geometric test for the parallelism between the regression direction and the first principal eigenvector. 
While this formulation shares the core philosophy of \cite{RefBFN99} by focusing on the direct geometric relationship between these directions, a straightforward extension of their Wald-type approach to the regression context suffers from a major methodological drawback; it requires the Moore--Penrose inverse of estimated matrices and thus relies on restrictive assumptions on the gaps among minor eigenvalues. 
To overcome this limitation, instead of the Wald-type formulation, we introduce a geometric test statistic that evaluates the Euclidean norm of the difference between the sample leading principal eigenvectors of two sum-of-squares matrices, which is equivalently formulated as an angle based on inner products. 
The asymptotic distribution of our statistic is derived as a weighted sum of chi-squared variables, avoiding the instability of eigenvector estimation under small minor eigenvalue gaps, a perspective that also underlies the test for partial common principal component subspaces of \citet{RefS03}.
Note that the model of partial common principal component subspaces considered in \cite{RefS99,RefS03} differs from the CPC subspace model considered in \cite{RefFNP95}.
While general reduced-rank regression~\citep[cf.][]{RefRVC22} focuses on imposing a low-rank structure on the regression coefficient matrix, the present study considers instead a structural constraint on the relationship between the regression direction and the principal component direction, and develops an inferential procedure for testing this relationship.
We show that the proposed test controls the asymptotic size and is consistent, and its finite-sample properties are examined through numerical experiments.
The performance of our method is illustrated through a reexamination of the classical painted turtle carapace data.

\section{Problem setting}\label{sec:2}
For $n$ specimens with $p$-dimensional observations denoted by $\bm{y}_1,\ldots,\bm{y}_n$, we consider the multivariate linear multiple regression model
\[
\bm{y}_i = \bm{\mu} + \bm{B}^\top \bm{x}_{n,i} + \bm{e}_i, \quad 
\E[\bm{e}_i] = \bm{0}_p, \quad
\Var[\bm{e}_i] = \bm{\Sigma}
\]
for $i=1,\ldots,n$, where $\bm{x}_{n,i}$ denotes the $q$-dimensional $(q<n)$ non-stochastic explanatory variable vector of the $i$-th specimen satisfying $\sum_{i=1}^n \bm{x}_{n,i} = \bm{0}_q$ and $\rk(\bm{X}_n)=q$ with
\[ \bm{X}_n = (\bm{x}_{n,1}, \ldots, \bm{x}_{n,n} )^\top. \]
Here, $\bm{\mu}$ is an unknown $p$-dimensional non-stochastic intercept vector; $\bm{B}$ is an unknown $q \times p$ non-stochastic coefficient matrix, for which $\rk(\bm{B}) = 1$ will be assumed later.
The unknown covariance matrix $\bm{\Sigma}$ is a $p\times p$ positive-definite matrix whose largest eigenvalue $\lambda_1$ is a unique root.
We do not assume that the minor eigenvalues $\lambda_2 \geq \lambda_3 \geq \cdots \geq \lambda_p > 0$ of $\bm{\Sigma}$ are distinct, and consider the spectral decomposition
\[ \bm{\Sigma} = \bm{\Gamma} \bm{\Lambda} \bm{\Gamma}^\top = (\bm{\gamma}_1,\ldots,\bm{\gamma}_p) \diag(\lambda_1,\ldots,\lambda_p) (\bm{\gamma}_1,\ldots,\bm{\gamma}_p)^\top , \]
where $\bm{\Gamma}$ is an orthogonal matrix.
This assumption is practically relevant, as minor eigenvalues are often close in biological data, leading to unstable eigenvector estimation.
Although the submatrix $(\bm{\gamma}_2,\ldots,\bm{\gamma}_p)$ is not uniquely determined, we fix a specific choice of the decomposition hereafter.
In this paper, we partly consider the setting where $\bm{\Sigma}$ depends on $n$, but we omit the subscript $n$ for simplicity, as our primary focus is the non-dependent case.
Henceforth, we assume that
\[ \bm{e}_1,\ldots,\bm{e}_n \stackrel{\mathrm{i.i.d.}}{\sim} \mathrm{N}_p(\bm{0}, \bm{\Sigma}) . \]

Let
$\bm{Y}=(\bm{y}_1,\ldots,\bm{y}_n)^\top$.
To infer $\bm{B}$ and $\bm{\Sigma}$, we introduce the sum-of-squares matrix $\bm{S}_T$, the regression sum-of-squares matrix $\bm{S}_R$, and the residual sum-of-squares matrix $\bm{S}_E$, defined as
\[
\bm{S}_T = \bm{Y}^\top (\bm{I}_n - n^{-1} \bm{1}_n \bm{1}_n^\top ) \bm{Y}, \quad
\bm{S}_R = \hat{\bm{B}}^\top \bm{X}_n^\top \bm{X}_n \hat{\bm{B}} , \quad
\bm{S}_E = \bm{S}_T - \bm{S}_R,
\]
where $\hat{\bm{B}} =  (\bm{X}_n^\top \bm{X}_n)^{-1} \bm{X}_n^\top \bm{Y}$, which is the dimension-wise ordinary least squares estimator of $\bm{B}$.
It follows that 
\[ \bm{S}_R = \bm{Y}^\top \bm{X}_n (\bm{X}_n^\top \bm{X}_n)^{-1} \bm{X}_n^\top \bm{Y}, \quad
\bm{S}_E = \bm{Y}^\top \{ \bm{I}_n - n^{-1} \bm{1}_n \bm{1}_n^\top - \bm{X}_n (\bm{X}_n^\top \bm{X}_n)^{-1} \bm{X}_n^\top \} \bm{Y}. \]
The multivariate allometric regression model, which is the primary focus of this paper, is a multivariate linear regression model under the restriction $\bm{B} = \bm{\alpha} \bm{\gamma}_1^\top$, where $\bm{\alpha}$ is a $q$-dimensional non-stochastic vector.
This restriction enforces a rank-one structure on $\bm{B}$.
That is to say, the multivariate allometric regression model considered in this paper is formulated as
\[
\bm{y}_i = \bm{\mu} + (\bm{\alpha}^\top \bm{x}_{n,i}) \bm{\gamma}_1 + \bm{e}_i \quad (i \in \{1,\ldots,n\}), \quad
(\bm{e}_1^\top,\ldots,\bm{e}_n^\top)^\top \sim \mathrm{N}_{np}(\bm{0}_{np} , \bm{\Sigma} \otimes \bm{I}_n) ,
\]
where $\otimes$ denotes the Kronecker product.

Under the allometric regression model, the directions of the eigenvectors of 
\[ \E[\bm{S}_T] = (n-1) \bm{\Sigma} + \| \bm{X}_n \bm{\alpha} \|^2 \bm{\gamma}_1 \bm{\gamma}_1^\top, \quad
\E[\bm{S}_R] = q \bm{\Sigma} + \| \bm{X}_n \bm{\alpha} \|^2 \bm{\gamma}_1 \bm{\gamma}_1^\top,\] 
and 
\[ \E[\bm{S}_E] = (n-1-q) \bm{\Sigma}\]
coincide with those of $\bm{\Sigma}$.
Hence, PCA can be conducted by utilizing not only $\bm{S}_E$ but also $\bm{S}_R$ or $\bm{S}_T$.
More generally, we can consider a hybrid PCA framework based on the convex combination
\[
\bm{S}(w) = (1-w) \bm{S}_R + w \bm{S}_E, \quad w \in [0,1].
\]
In particular, let $\hat{\bm{\gamma}}_1 (w)$ be the first principal eigenvector of $\bm{S}(w)$ for $w \in [0,1]$, which we refer to as a hybrid estimator.
This hybrid construction incorporates both covariance and mean structure and plays a key role in stabilizing the estimation of $\bm{\gamma}_1$.
Throughout this paper, to avoid the sign indeterminacy of the sample principal eigenvectors, we assume that
\[ \hat{\bm{\gamma}}_1(w)^\top \bm{\gamma}_1 \geq 0 \]
for all $w \in [0,1]$.

For this hybrid PCA framework, the consistency of $\hat{\bm{\gamma}}_1 (w)$ with $w \in [0,1]$ has been established under several asymptotic regimes in \citet{RefTM25}.
In Section~\ref{sec:3}, we prove the asymptotic normality of $\hat{\bm{\gamma}}_1 (w)$ for $w \in [0,1]$ under the classical asymptotic regime.
Based on this result, in Section~\ref{sec:4}, we propose a statistical test for the parallelism between the regression direction and the first principal eigenvector within the rank-one reduced-rank regression framework and investigate its theoretical properties.
Although the limiting distribution of the proposed test statistic is a weighted sum of chi-squared variables, we handle its practical computation efficiently via a moment-matching approximation, ensuring ease of implementation.
Moreover, finite-sample properties of the proposed test are verified through numerical experiments in Section~\ref{sec:5}.
In Section~\ref{sec:6}, we analyze the painted turtle carapace data to illustrate the practical behavior of the proposed test and confirm that the conclusions of \citet{RefBFN99} can be reproduced.
Furthermore, in Section~\ref{sec:7}, we complement these findings by analyzing the asymptotic normality under a weak identifiability regime where $\lambda_1 - \lambda_2 \to 0$.
Concluding remarks are given in Section~\ref{sec:8}.
Proofs of the theoretical results, detailed numerical results, and additional robustness checks under heteroscedasticity are provided in Appendix.

\begin{rem}
As stated in Section~\ref{sec:1}, the assumption $\lambda_1 > \lambda_2$ is natural in allometry studies.
In contrast, it is generally less clear whether the strict inequality holds for the minor eigenvalues, such as $\lambda_{p-1} > \lambda_p$.
Accordingly, we assume the ordered constraints $\lambda_1 > \lambda_2 \geq \lambda_3 \geq \cdots \geq \lambda_p > 0$.
\end{rem}

\begin{rem}
The assumption of non-stochastic explanatory variables is justified either as a standard conditional inference framework given the observed covariates or as a direct reflection of stratified sampling designs where group sample sizes are explicitly fixed.
\end{rem}

\section{Asymptotic normality of the first principal eigenvector}\label{sec:3}

In this section, we provide properties of the hybrid estimator $\hat{\bm{\gamma}}_1(w)$ for $w \in [0,1]$ under the multivariate allometric regression model.
To establish our main results, we first recall the following non-asymptotic result of \citet{RefTM25}.

\begin{lem} \label{lem1}
Let 
\[ 
a = \tr(\bm{\Sigma}^2) + (\tr (\bm{\Sigma}))^2 ,\quad 
b = \lambda_1 + \tr(\bm{\Sigma}) ,\quad
c= \| \bm{X}_n \bm{\alpha} \|^2, \quad
d= \lambda_1 - \lambda_2, \]
and $w \in [0, 1]$.
The ordered eigenvalues $\lambda_1(w),\ldots,\lambda_p(w)$ of $\E[\bm{S}(w)]$ are given by
\begin{align*}
&\lambda_1(w) = \{q(1-w) + (n-1-q)w \} \lambda_1 + (1-w) c, \\
&\lambda_k(w) = \{q(1-w) + (n-1-q)w \} \lambda_k 
\end{align*}
for $k \in  \{2,\ldots,p\}$.
It holds that
\begin{align*}
& \E[ \| \bm{S}(w) - \E[\bm{S}(w)] \|_{\mathrm{F}}^2]
= \{q(1-2w) + (n-1)w^2 \} a + 2 (1-w)^2 bc, \\
& \E[ \| \hat{\bm{\gamma}}_1(w) - \bm{\gamma}_1 \|^2 ] \leq \frac{ 8a \{q(1-2w) + (n-1)w^2 \} + 16bc(1-w)^2 }{[d\{ q + (n-1-2q) w \} + c (1-w)]^2},
\end{align*}
where $\|\cdot\|_{\mathrm{F}}$ denotes the Frobenius norm.
\end{lem}

Using Lemma~\ref{lem1}, we discuss asymptotic properties of $\hat{\bm{\gamma}}_1(w)$ for $w \in [0,1]$.
Consider the following classical asymptotic regime.

\begin{ass}\label{ass1}
As $n\to\infty$ with $p,q,\bm{\alpha},\bm{\Sigma}$ fixed, $\bm{X}_n \bm{\alpha} \neq \bm{0}$ and $n^{-1} \bm{X}_n^\top \bm{X}_n$ converges to a positive-definite matrix.
\end{ass}

Under Assumption~\ref{ass1}, $c/n = \bm{\alpha}^\top ( n^{-1} \bm{X}_n^\top  \bm{X}_n ) \bm{\alpha}$ converges, and we define
\begin{equation} \label{cinf}
c_\infty := \lim_{n \to \infty}  (n^{-1}  \|\bm{X}_n \bm{\alpha} \|^2).    
\end{equation}

\begin{thm}\label{thm1}
Let Assumption~\ref{ass1} hold.
For $w \in [0,1]$, as $n\to\infty$, 
\begin{align*}
& n^{1/2} (\hat{\bm{\gamma}}_1(w) - \bm{\gamma}_1) \\
& \Rightarrow
\mathrm{N}_p \left( \bm{0}_p , \{w^2 \lambda_1 + c_\infty (1-w)^2 \} \sum_{k=2}^p \frac{  \lambda_k}{\{ w (\lambda_1 - \lambda_k) + (1-w) c_\infty \}^2} \bm{\gamma}_k \bm{\gamma}_k^\top \right),
\end{align*}
where $\Rightarrow$ denotes convergence in distribution.
\end{thm}

\begin{cor}\label{cor0}
Let Assumption~\ref{ass1} hold.
For $w \in [0,1]$, as $n\to\infty$, 
\[
n \E[\| \hat{\bm{\gamma}}_1(w) - \bm{\gamma}_1 \|^2]
\to \{w^2 \lambda_1 + c_\infty (1-w)^2 \} \sum_{k=2}^p \frac{  \lambda_k}{\{ w (\lambda_1 - \lambda_k) + (1-w) c_\infty \}^2}.
\]
\end{cor}

\begin{rem}
It follows from Theorem~\ref{thm1} that, as long as $w \in [0,1)$, the asymptotic variance remains bounded even when $\lambda_1-\lambda_2$ is small.
Corollary~\ref{cor0} implies that the same conclusion holds for the mean squared error (MSE).
\end{rem}

\begin{rem}
Using Corollary~\ref{cor0}, we can formally choose a data-dependent weight $w$ that minimizes the plug-in estimator of the asymptotic MSE. 
This suggests a possible data-adaptive improvement.
Although implementing this selection rule and proving its asymptotic optimality under the $M$-estimation framework is an interesting direction, it is beyond the scope of the present study and will be investigated in future work.
\end{rem}

From Theorem~\ref{thm1}, the following corollaries follow.

\begin{cor}\label{cor1}
Let Assumption~\ref{ass1} hold.
Then, as $n\to\infty$, 
\[
n^{1/2} (\hat{\bm{\gamma}}_1(1) - \hat{\bm{\gamma}}_1(0))
\Rightarrow
\mathrm{N}_p \left( \bm{0}_p , \sum_{k=2}^p \left\{ \frac{  \lambda_1}{ (\lambda_1 - \lambda_k)^2}  + \frac{1}{c_\infty} \right\} \lambda_k \bm{\gamma}_k \bm{\gamma}_k^\top \right).
\]
\end{cor}

\begin{cor}\label{cor2}
Let Assumption~\ref{ass1} hold.
Then, as $n\to\infty$, 
\[
n\{ 1 - (\hat{\bm{\gamma}}_1(1)^\top \hat{\bm{\gamma}}_1(0))^2 \}
\Rightarrow
\sum_{k=2}^p \left\{ \frac{  \lambda_1}{ (\lambda_1 - \lambda_k)^2}  + \frac{1}{c_\infty} \right\} \lambda_k  \chi^2_{1,k},
\]
where $\chi^2_{1,2},\ldots,\chi^2_{1,p}$ are independent $\chi^2(1)$ variables.
\end{cor}

\section{Test of parallelism between the regression direction and the first principal eigenvector}\label{sec:4}

In this section, we consider the reduced-rank regression model of rank 1 formulated by
\[
\bm{y}_i = \bm{\mu} + (\bm{\alpha}^\top \bm{x}_{n,i}) \bm{\beta}  + \bm{e}_i \quad (i \in \{1,\ldots,n\}), \quad
(\bm{e}_1^\top,\ldots,\bm{e}_n^\top)^\top \sim \mathrm{N}_{np}(\bm{0}, \bm{\Sigma} \otimes \bm{I}_n) ,
\]
where $\bm{\beta} \in \R^p$ satisfies $\|\bm{\beta} \|=1$ and $\bm{\beta}^\top \bm{\gamma}_1 \geq 0$ to fix the sign indeterminacy.
For the general theory of reduced-rank regression, we refer to, for example, \cite{RefRVC22}.
We test null hypothesis $\Hc_0: \bm{\beta}^\top \bm{\gamma}_1 = 1$ against alternative $\Hc_1: \bm{\beta}^\top \bm{\gamma}_1 < 1$ at a significance level $\alpha_0$.
The null hypothesis $\bm{\beta}^\top \bm{\gamma}_1 = 1$ is equivalent to $\bm{\beta} \parallel \bm{\gamma}_1$.
We treat $\Hc_1$ as a fixed alternative, so that $\bm{\beta}$ and $\bm{\gamma}_1$ do not depend on $n$.

\begin{rem}
Under $\Hc_0$, since $\bm{\beta}$ lies on the unit sphere with a fixed sign, no boundary constraints arise in the asymptotic inference.
\end{rem}

\begin{rem}
Our setting corresponds to that of \cite{RefBFN99} in the regression context, under an additional assumption of common covariance matrices across two populations.
\end{rem}

Let $\hat{\lambda}_1,\ldots,\hat{\lambda}_p$ be the ordered eigenvalues of
\[ 
\hat{\bm{\Sigma}} = (n-q-1)^{-1} \bm{S}_E ,
\] 
which are estimators for the ordered eigenvalues $\lambda_1,\ldots,\lambda_p$ of $\bm{\Sigma}$.
Moreover, let
\[
\hat{c}_\infty 
= n^{-1} (\tr(\bm{S}_R) - q \tr(\hat{\bm{\Sigma}}) )
\]
be an estimator of $c_\infty$ defined in \eqref{cinf}.
Using these, we consider a statistic
\begin{equation} \label{Vhat}
\hat{V}_k = \left\{ \frac{  \hat{\lambda}_1}{ (\hat{\lambda}_1 - \hat{\lambda}_k)^2}  + \frac{1}{\hat{c}_\infty} \right\} \hat{\lambda}_k
\end{equation}
for  $k \in \{2,\ldots,p\}$.
The validity of the plug-in approach is justified by the following lemma.

\begin{lem}\label{lem2}
Let Assumption~\ref{ass1} hold.
Then, 
\[
\sum_{k=2}^p \hat{V}_k \chi^2_{1,k}
-
\sum_{k=2}^p \left\{ \frac{  \lambda_1}{ (\lambda_1 - \lambda_k)^2}  + \frac{1}{c_\infty} \right\} \lambda_k  \chi^2_{1,k}
\stackrel{\mathrm{p}}{\to} 0
\]
as $n\to\infty$, where $\stackrel{\mathrm{p}}{\to}$ denotes convergence in probability.
\end{lem}

Let $q(\alpha_0)$ and $\hat{q}(\alpha_0)$ be the upper $\alpha_0$-quantile of the distribution of
\[
\sum_{k=2}^p \left\{ \frac{  \lambda_1}{ (\lambda_1 - \lambda_k)^2}  + \frac{1}{c_\infty} \right\} \lambda_k  \chi^2_{1,k} \quad \text{and} \quad
\sum_{k=2}^p \hat{V}_k \chi^2_{1,k},
\]
respectively.
Based on a test statistic
\[ T_n = n \|\hat{\bm{\gamma}}_1(1) - \hat{\bm{\gamma}}_1(0) \|^2 , \]
consider a test function
\[
\varphi_n = 1\{ T_n > \hat{q} (\alpha_0) \} ,
\]
where $1\{ \cdot \}$ denotes the indicator function.
The following theorem shows that the test $\varphi_n$ asymptotically controls the size and is consistent.

\begin{thm}\label{thm3}
Let Assumption~\ref{ass1} hold.
\begin{itemize}
\item[(i)]
Under $\Hc_0$, $\lim_{n\to\infty} \Pr(\varphi_n=1) = \alpha_0$.
\item[(ii)]
Under $\Hc_1$, $\lim_{n\to\infty} \Pr(\varphi_n=1) = 1$.
\end{itemize}
\end{thm}

\begin{rem}
In light of Corollary~\ref{cor2}, the proposed test statistic evaluates the geometric discrepancy between two directions by utilizing two sum-of-squares matrices.
Although it is formally possible to construct a Wald-type test statistic using the Moore--Penrose inverse of the estimated asymptotic covariance matrix as in \cite{RefBFN99}, such an approach faces practical challenges in our setting.
Specifically, since the gaps among the minor eigenvalues $\lambda_2,\ldots,\lambda_p$ are not assumed, the individual eigenvectors $\bm{\gamma}_2,\ldots,\bm{\gamma}_p$ are not necessarily identifiable, which can make the direct plug-in estimation of the asymptotic covariance matrix unstable.
This motivates our use of the proposed statistic, which avoids the direct estimation of eigenvectors $\bm{\gamma}_2,\ldots,\bm{\gamma}_p$.
A similar motivation, namely avoiding the instability of eigenvectors under small gaps among the minor eigenvalues, appears in the test for partial common principal component subspaces of \citet{RefS03}, where the test statistic is also constructed so that its asymptotic distribution is given by a weighted sum of chi-squared variables.
\end{rem}

To calculate $\hat{q}(\alpha_0)$, we use a moment-matching method based on the first two moments, which provides a computationally efficient approximation.
Other approximations for linear combinations of chi-squared variables are available but not pursued here.
Let
\[
K = \sum_{k=2}^p \left\{ \frac{  \lambda_1}{ (\lambda_1 - \lambda_k)^2}  + \frac{1}{c_\infty} \right\} \lambda_k  \chi^2_{1,k}.
\]
We approximate the distribution of $K$ by matching its first two moments with a scaled chi-squared distribution.
The mean and variance of $K$ are given by
\[
\E[K] = \sum_{k=2}^p \left\{ \frac{  \lambda_1}{ (\lambda_1 - \lambda_k)^2}  + \frac{1}{c_\infty} \right\} \lambda_k 
\quad \text{and} \quad
\Var[K] = 2 \sum_{k=2}^p \left\{ \frac{  \lambda_1}{ (\lambda_1 - \lambda_k)^2}  + \frac{1}{c_\infty} \right\}^2 \lambda_k^2 ,
\]
respectively.
Let $\tilde{K}$ be a random variable satisfying $\tilde{K}/\kappa \sim \chi^2(\phi)$, where
\[
\kappa =  \frac{\Var[K]}{2 \E[K]}, \quad \phi = \frac{2(\E[K])^2}{\Var[K]}.
\]
By construction, we have $\E[\tilde{K}] = \E[K]$ and $\Var[\tilde{K}] = \Var[K]$.
Thus, we approximate the distribution of $K$ by that of $\tilde{K}$, implying that $q(\alpha_0)$ is approximated by $\kappa$ multiplied by the upper $\alpha_0$-quantile of $\chi^2(\phi)$.
In practice, as $\kappa$ and $\phi$ are unknown, we construct their plug-in estimators by utilizing $\hat{V}_2, \ldots, \hat{V}_p$ defined in \eqref{Vhat} and letting
\[
\hat{W}_k = \left\{ \frac{ \hat{\lambda}_1}{ (\hat{\lambda}_1 - \hat{\lambda}_k)^2}  + \frac{1}{\hat{c}_\infty} \right\}^2 \hat{\lambda}_k^2  \quad (k \in \{2, \ldots, p\}).
\]
The plug-in estimators for $\kappa$ and $\phi$ are given by
\[
\hat{\kappa} = \frac{ \sum_{k=2}^p \hat{W}_k }{\sum_{k=2}^p \hat{V}_k}, \quad
\hat{\phi} = \frac{(\sum_{k=2}^p \hat{V}_k)^2}{\sum_{k=2}^p \hat{W}_k }.
\]
When $p=2$, $\hat{\phi} = 1$, so the moment-matching approximation reduces to the plug-in approximation.
The approximate $P$-value is calculated based on the $\chi^2$ approximation for the distribution of $T_n/\hat{\kappa}$, where the degrees of freedom $\hat{\phi}$ and the scaling factor $\hat{\kappa}$ are evaluated at their realized values.

\section{Numerical experiments}\label{sec:5}
In this section, we evaluate the performance of the proposed test $\varphi_n$ through numerical experiments.
For comparison, we consider a Wald-type test as an alternative approach.
Let
\[
\tilde{T}_n = n (\hat{\bm{\gamma}}_1(1) - \hat{\bm{\gamma}}_1(0))^\top \hat{\bm{\Upsilon}}^{+} (\hat{\bm{\gamma}}_1(1) - \hat{\bm{\gamma}}_1(0)) ,
\]
where $\hat{\bm{\Upsilon}}^+$ is the Moore--Penrose inverse of
\[
\hat{\bm{\Upsilon}} =  \sum_{k=2}^p \left\{ \frac{  \hat{\lambda}_1}{ ( \hat{\lambda}_1 - \hat{\lambda}_k)^2}  + \frac{1}{\hat{c}_\infty} \right\} \hat{\lambda}_k \hat{\bm{\gamma}}_k \hat{\bm{\gamma}}_k^\top
\]
and $\hat{\bm{\gamma}}_2,\ldots,\hat{\bm{\gamma}}_p$ are the second to $p$-th eigenvectors of $\bm{S}_T$.
Although $\E[\bm{S}_E]$ and $\E[\bm{S}_T]$ share the same eigenvectors under $\Hc_0$, we employ $\bm{S}_T$ because it has larger degrees of freedom.
We then define a test
\[
\tilde{\varphi}_n = 1\{ \tilde{T}_n > q_{\chi^2(p-1)}(\alpha_0) \} ,
\]
where $q_{\chi^2(p-1)}(\alpha_0)$ is the upper $\alpha_0
$-quantile of the chi-squared distribution with $(p-1)$ degrees of freedom.
When $p=2$, the tests $\tilde{\varphi}_n$ and $\varphi_n$ coincide.

\begin{rem}
If the minor eigenvalues satisfy $\lambda_2 > \ldots > \lambda_p$, then $\hat{\bm{\gamma}}_k \stackrel{\mathrm{p}}{\to} \bm{\gamma}_k$ holds for all $k = 2, \ldots, p$, and hence
\[ 
\hat{\bm{\Upsilon}} \stackrel{\mathrm{p}}{\to} \sum_{k=2}^p \left\{ \frac{\lambda_1}{(\lambda_1 - \lambda_k)^2} + \frac{1}{c_\infty} \right\} \lambda_k \bm{\gamma}_k \bm{\gamma}_k^\top .
\]
However, if any of the minor eigenvalues coincide, individual eigenvectors $\bm{\gamma}_2, \ldots, \bm{\gamma}_p$ are no longer identifiable.
When they are tightly clustered, with finite $n$, the estimated eigenvectors may exhibit substantial variability within the corresponding eigenspace, and can occasionally align in nearly the same direction as each other.
Consequently, the Moore--Penrose inverse of $\hat{\bm{\Upsilon}}$ may become unstable in finite sample settings, which can adversely affect the behavior of the Wald-type statistic. 
\end{rem}

The simulation settings, using the notation of Section~\ref{sec:4}, are as follows:
$\alpha_0=0.05$, $n \in \{ 20,50,100,200,500 \}$, $p=5$ (fixed), $q=1$ (fixed), $\bm{\mu}=\bm{0}_5$ (fixed), $c/n = n^{-1} \| \bm{X}_n \bm{\alpha} \|^2 \in \{2, 10\}$, $\bm{\Sigma} = \diag(\bm{\lambda})$ with
\begin{align*} 
\bm{\lambda}^\top 
\in \{ (10,8,6,4,2),
(10,1,1,1,1), 
(2,1.5,1.5,1.5,1), 
(1.8,1,1,1,1), 
(1.2,1,1,1,1) \},
\end{align*}
and $\bm{\beta} = (1,0,\ldots,0)^\top$ for $\Hc_0$ or $\bm{\beta} = (\cos \theta, \sin \theta, 0, \ldots, 0)^\top$ with $\theta \in \{\pi/4, \pi/2\}$ for $\Hc_1$, where  $\bm{X}_n = (x_{n,1,1},\ldots, x_{n,1,n} )^\top$ represents either a binary or a continuous explanatory variable, depending on the setting.
In the binary case, $x_{n,1,i} = 1/2$ for $i \in \{1,\ldots,n/2\}$ and $x_{n,1,i} = -1/2$ for $i \in \{n/2 + 1,\ldots,n\}$, whereas in the continuous case, $\bm{X}_{n}$ is taken to be an equally spaced sequence on $[-1,1]$.
The considered configurations of $\bm{\lambda}$ range from well-separated eigenvalues to nearly indistinguishable ones.
The simulation is repeated 1000 times.
For the binary case, we compare $\varphi_n$, $\tilde{\varphi}_n$, and $\varphi_{\text{BFN}}$.
For the continuous case, we compare $\varphi_n$ and $\tilde{\varphi}_n$.

The empirical type I error rates are summarized in Figures~\ref{fig:size_1}--\ref{fig:size_2}, while the detailed numerical results are provided in Appendix~\ref{sec:A2}.
The main observations are as follows.
\begin{itemize}
\item 
For both $\varphi_n$ and $\tilde{\varphi}_n$, the results are qualitatively similar between the discrete and continuous cases.

\item 
The type I error rates of all methods approach $\alpha_0$ as the sample size increases, in accordance with the asymptotic theory. 
In particular, for $n \geq 50$, the type I error rates of $\varphi_n$ are close to $\alpha_0 = 0.05$, except when $\bm{\lambda}^\top = (2,1.5,1.5,1.5,1)$ or $(1.2,1,1,1,1)$, where the gaps $\lambda_1 - \lambda_2$ are considered to be not large enough.

\item 
The type I error rates of $\varphi_n$ are comparable to or more conservative than those of $\tilde{\varphi}_n$ in many cases. 
In particular, $\varphi_n$ provides more reliable control of the type I error when the minor eigenvalue gaps are small, as suggested by the theoretical properties.

\item 
In the binary case, both $\varphi_n$ and $\tilde{\varphi}_n$ are more conservative than $\varphi_{\text{BFN}}$.
\end{itemize}
Similar tendencies are also observed for the rejection rates under the alternative hypotheses; see the tables in Appendix~\ref{sec:A2}.
Specifically, in many cases, $\varphi_n$ tends to have the lowest rejection rate, followed by $\tilde{\varphi}_n$, whereas $\varphi_{\text{BFN}}$ tends to have the highest.
However, since $\varphi_{\text{BFN}}$ substantially exceeds the nominal significance level in many cases, these comparisons should be interpreted with caution.

\begin{figure}[!tbp]
\centering
\includegraphics[width=.9\linewidth]{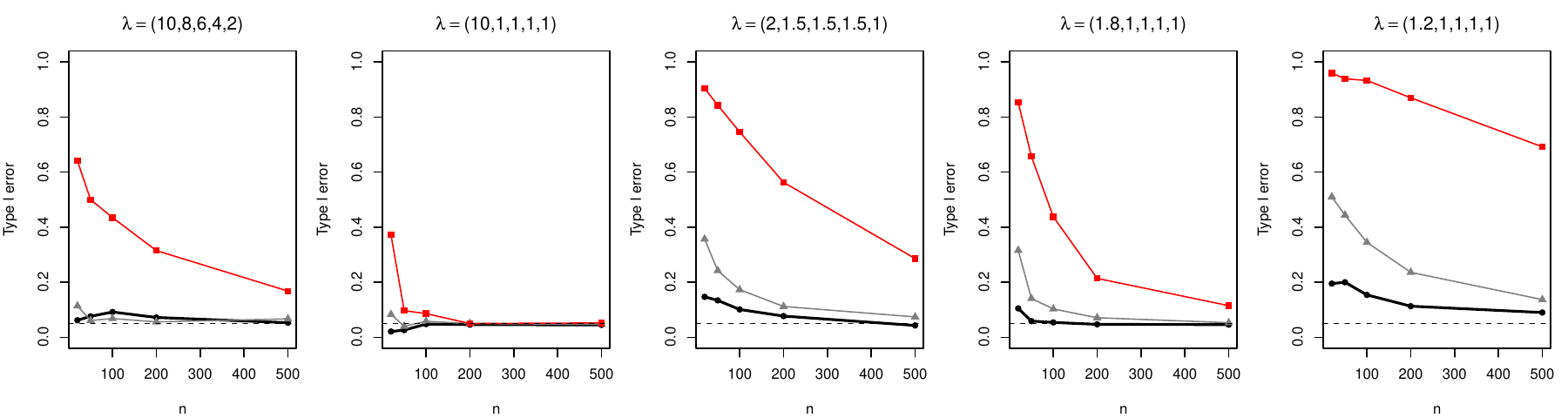}
\vspace{1em}
\includegraphics[width=.9\linewidth]{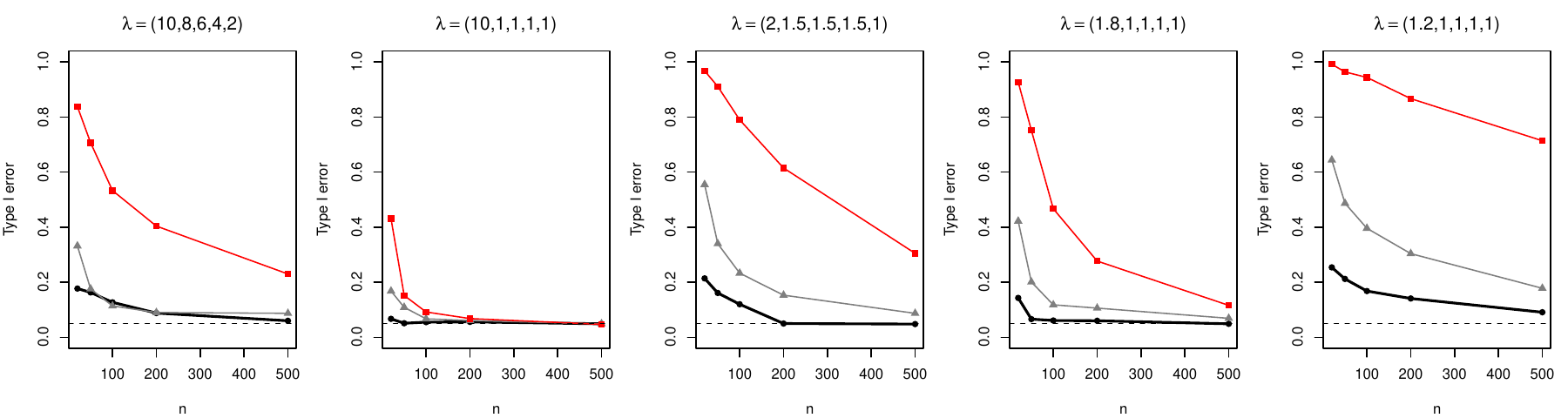}
\caption{Empirical type I error rates for the binary case with $n^{-1} \| \bm{X}_n \bm{\alpha} \|^2 = 2$ (top) and $10$ (bottom). Black circles: $\varphi_n$; gray triangles: $\tilde{\varphi}_n$; red squares: $\varphi_{\text{BFN}}$; dashed line: 0.05.}
\label{fig:size_1}
\end{figure}

\begin{figure}[!tbp]
\centering
\includegraphics[width=.9\linewidth]{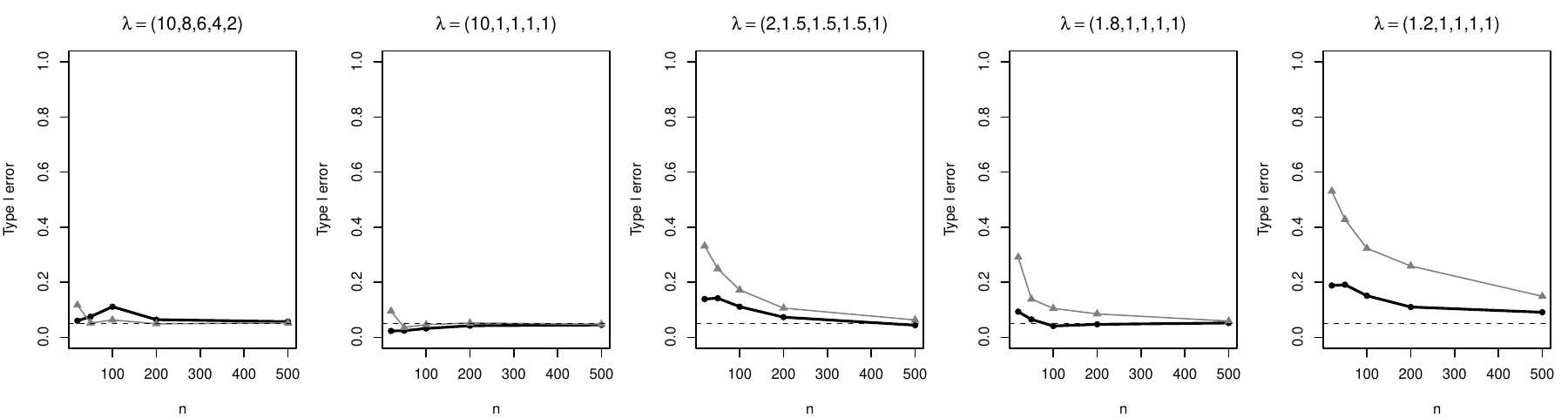}
\vspace{1em}
\includegraphics[width=.9\linewidth]{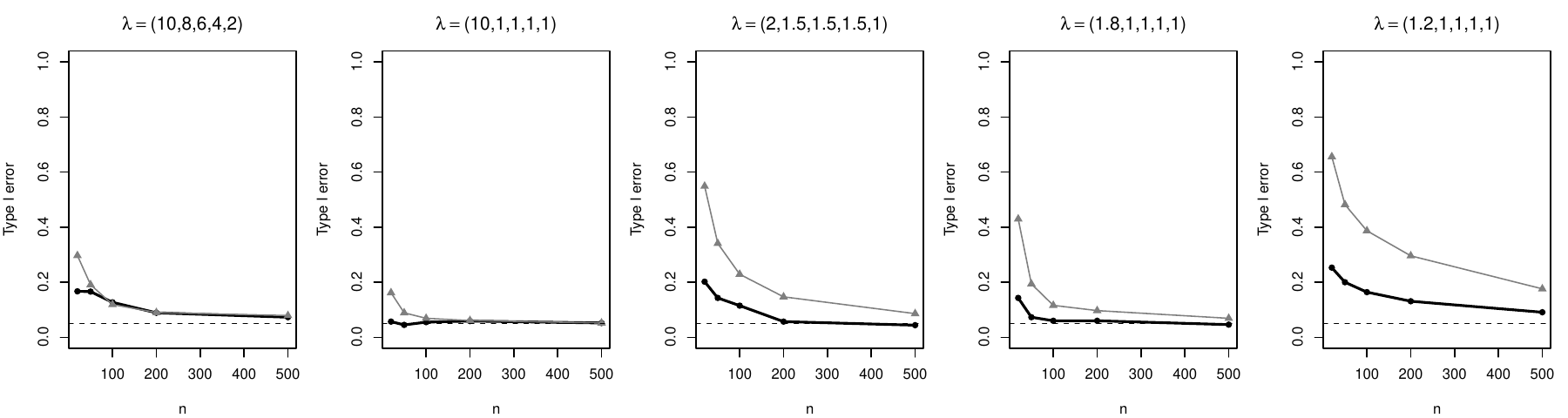}
\caption{Empirical type I error rates for the continuous case with $n^{-1} \| \bm{X}_n \bm{\alpha} \|^2 = 2$ (top) and $10$ (bottom). Black circles: $\varphi_n$; gray triangles: $\tilde{\varphi}_n$; dashed line: 0.05.}
\label{fig:size_2}
\end{figure}

These results suggest that the previous conclusions based on $\varphi_{\text{BFN}}$ in the turtle carapace analysis may be affected by inflated type I error rates. 
This motivates a reexamination of the dataset using the proposed method in the next section.

\begin{rem}
Additional numerical results under heteroscedasticity are provided in Appendix~\ref{sec:A3}.
\end{rem}

\section{Painted turtle carapace data}\label{sec:6}

In this section, we analyze the painted turtle carapace data  \citep{RefJM60} using the multivariate allometric regression model with our proposed method.
The dataset is available as the \texttt{tortues} data in the \textsc{R} package \texttt{ade4}.
With a slight abuse of notation, random variables and their realized values are not notationally distinguished throughout this section.

The response vector $\bm{y} = (y_1,y_2,y_3)^\top$ consists of the transformed carapace dimensions:
\[ 
y_1 = 10 \log(\text{length}), \quad 
y_2 = 10 \log(\text{width}), \quad 
y_3 = 10 \log(\text{height}), 
\]
and the explanatory variable $x$ is defined as an indicator variable for sex, taking the value $0$ for females and $1$ for males.
Following \citet{RefBFN99}, the ten-times logarithmic transformation is adopted for all structural measurements.
The data comprise $n_F = 24$ female and $n_M=24$ male turtles.
For $x = 0$ (female) and $x=1$ (male), the sample mean vectors of the response variables, denoted by $\bar{\bm{y}}_F$ and $\bar{\bm{y}}_M$, and the corresponding normalized difference vector $\bm{d} = \|\bar{\bm{y}}_F - \bar{\bm{y}}_M\|^{-1} (\bar{\bm{y}}_F - \bar{\bm{y}}_M)$ are given by
\[
\bar{\bm{y}}_F =
\begin{bmatrix}
49.0036 \\
46.2412 \\
39.2916
\end{bmatrix}, \quad
\bar{\bm{y}}_M =
\begin{bmatrix}
47.2544 \\
44.7757 \\
37.0422
\end{bmatrix}, \quad 
\bm{d} = 
\begin{bmatrix}
0.5459 \\
0.4573 \\
0.7020
\end{bmatrix}.
\]
The sample covariance matrices $\hat{\bm{\Sigma}}_F$ and $\hat{\bm{\Sigma}}_M$ (with the divisor $n_\cdot -1 = 23$) and their spectral decompositions are given by
\begin{align*}
&\hat{\bm{\Sigma}}_F  
= \begin{bmatrix}
2.6391 & 2.0111 & 2.5538 \\
2.0111 & 1.6231 & 1.9491 \\
2.5538 & 1.9491 & 2.8005
\end{bmatrix} = \hat{\bm{\Gamma}}_F \hat{\bm{\Lambda}}_F \hat{\bm{\Gamma}}_F^\top, \\
&\hat{\bm{\Gamma}}_F =\begin{bmatrix}
0.6164 & -0.3973 & \hphantom{-}0.6799 \\
0.4763 & -0.4994 & -0.7237 \\
0.6271 & \hphantom{-}0.7699 & -0.1185
\end{bmatrix}, \quad
\hat{\bm{\Lambda}}_F = 
\diag\left(\begin{bmatrix}
6.7913 \\
0.2183 \\
0.0530
\end{bmatrix}  \right)
\end{align*}
and
\begin{align*}
&\hat{\bm{\Sigma}}_M
= \begin{bmatrix}
1.1072 & 0.8019 & 0.8162 \\
0.8019 & 0.6417 & 0.6005 \\
0.8162 & 0.6005 & 0.6767
\end{bmatrix}  = \hat{\bm{\Gamma}}_M \hat{\bm{\Lambda}}_M \hat{\bm{\Gamma}}_M^\top,\\
&\hat{\bm{\Gamma}}_M = 
\begin{bmatrix}
0.6832 & -0.1512 & \hphantom{-}0.7144 \\
0.5102 & -0.6011 & -0.6151 \\
0.5225 & \hphantom{-}0.7847 & -0.3335
\end{bmatrix}, \quad
\hat{\bm{\Lambda}}_M =
\diag\left(\begin{bmatrix}
2.3303 \\
0.0595 \\
0.0358
\end{bmatrix}  \right),
\end{align*}
respectively.
As in \citet{RefJ63}, the first principal eigenvectors for both sexes are understood to represent the size factor.
These values differ from those reported by \citet{RefBFN99}, which is likely attributable to differences in data formatting or rounding errors in the computational processes.
For both sexes, the estimated eigenvalues exhibit a prominent dominance of the first component, such that $\hat{\lambda}_1 \gg \hat{\lambda}_2 > \hat{\lambda}_3$. 
Indeed, the contribution ratios of the first principal component exceed 0.95 for both female and male turtles.
In particular, for male turtles, there is a characteristic $\hat{\lambda}_2 \approx \hat{\lambda}_3$.
The leading principal eigenvectors of $\hat{\bm{\Sigma}}_F$ and $\hat{\bm{\Sigma}}_M$ are denoted by $\hat{\bm{\gamma}}_{F,1}$ and $\hat{\bm{\gamma}}_{M,1}$, respectively.
The mutual inner product between these two eigenvectors and their respective inner products with $\bm{d}$ are
\[
\hat{\bm{\gamma}}_{F,1}^\top \hat{\bm{\gamma}}_{M,1} = 0.9917, \quad
\bm{d}^\top \hat{\bm{\gamma}}_{F,1} = 0.9945, \quad
\bm{d}^\top \hat{\bm{\gamma}}_{M,1} = 0.9731.
\]
Although the inner product $\bm{d}^\top \hat{\bm{\gamma}}_{M,1}$ is the smallest among these values, all three metrics are close to unity, indicating a high degree of collinearity among the vectors.

By applying the hybrid PCA framework, the estimated principal eigenvectors for $w=0, 1/2$, and $1$ are
\[
\hat{\bm{\gamma}}_1(0) = \begin{bmatrix}
0.5459 \\
0.4573 \\
0.7020
\end{bmatrix}, \quad
\hat{\bm{\gamma}}_1(1/2) = 
\begin{bmatrix}
0.6025 \\
0.4759 \\
0.6407
\end{bmatrix}, \quad
\hat{\bm{\gamma}}_1(1) = 
\begin{bmatrix}
0.6345 \\
0.4858 \\
0.6012
\end{bmatrix}.
\]
By definition, $\hat{\bm{\gamma}}_1(0) = \bm{d}$. 
Furthermore, due to $n_F = n_M$, the squared distance between the mean vectors satisfies the exact algebraic relation $\|\bar{\bm{y}}_F - \bar{\bm{y}}_M \|^2 = (4/n) \tr(\hat{\bm{\Sigma}}_R)$.
Notably, $\hat{\bm{\gamma}}_1(0)^\top \hat{\bm{\gamma}}_1(1) = 0.9906$, which is close to unity.
This indicates that the regression direction and the principal component directions are nearly aligned.
Moreover, the estimated covariance matrix and its spectral decomposition are given by
\begin{align*}
&\hat{\bm{\Sigma}} = 
\begin{bmatrix}
1.8732 & 1.4065 & 1.6850 \\
1.4065 & 1.1324 & 1.2748 \\
1.6850 & 1.2748 & 1.7386
\end{bmatrix} 
= \hat{\bm{\Gamma}} \hat{\bm{\Lambda}} \hat{\bm{\Gamma}}^\top\\
& \hat{\bm{\Gamma}} =
\begin{bmatrix}
0.6345 & -0.3826 & \hphantom{-}0.6716 \\
0.4858 & -0.4785 & -0.7315 \\
0.6012 & \hphantom{-}0.7904 & -0.1177
\end{bmatrix}, \quad 
\hat{\bm{\Lambda}} = 
\diag\left(\begin{bmatrix}
4.5469 \\
0.1513 \\
0.0460
\end{bmatrix}  \right),
\end{align*}
from which we observe that the components of $\hat{\bm{\gamma}}_1(1)$ fall within the respective intervals spanned by the corresponding components of $\hat{\bm{\gamma}}_{F,1}$ and $\hat{\bm{\gamma}}_{M,1}$. 
Consequently, $\hat{\bm{\gamma}}_1(0)^\top \hat{\bm{\gamma}}_1(1) = \bm{d}^\top \hat{\bm{\gamma}}_1(1)$ also lies between the values of $\bm{d}^\top \hat{\bm{\gamma}}_{F,1}$ and $\bm{d}^\top \hat{\bm{\gamma}}_{M,1}$.
The test statistics are
\[
T_n = 11.1155\hat{\kappa}, \quad
\tilde{T}_n = 12.2423 ,
\]
with the corresponding approximate $P$-values being $0.0021$ and $0.0022$, respectively, where $\hat{\kappa} = 0.0812$ and $\hat{\phi}= 1.5491$. 
Consequently, the parallelism hypothesis $\mathcal{H}_0$ is rejected in either test.
This rejection, despite the large inner product $\hat{\bm{\gamma}}_1(0)^\top \hat{\bm{\gamma}}_1(1) = 0.9906$, suggests that even small departures from parallelism may be statistically detectable in this setting.
This can be attributed to the fact that $\hat{\lambda}_2=0.1513$ and $\hat{\lambda}_3=0.0460$ are small relative to $\hat{\lambda}_1 = 4.5469$ and $\hat{c}_\infty = 2.4679$, leading to small standard errors of $\hat{\bm{\gamma}}_1(0)$ and $\hat{\bm{\gamma}}_1(1)$.
Consistent with \citet{RefBFN99}, we find evidence to reject the parallelism; however, this statistical rejection may not necessarily imply that the departure is biologically meaningful.

To assess the above results in small sample situations, we also conduct auxiliary analyses using two sets of systematically subsampled datasets: one retaining only specimens with even or odd indices ($n=24$, that is, 12 of each sex), and the other using the three residue classes modulo three ($n=16$, that is, 8 of each sex). 
The results are summarized in Table~\ref{tab:subsampling}.
These results show that the test conclusions may vary substantially under systematic subsampling with small sizes, highlighting the sensitivity of the inference in regimes with small minor eigenvalues.
This variation reflects that higher inner products tend to yield smaller test statistics, whereas slight deviations lead to larger values.

\begin{table}[tbp]
\caption{Analyses under systematic subsampling of the painted turtle carapace data.}
\label{tab:subsampling}
\centering
\begin{small}
\begin{tabular}{lcccccc}
\hline
Subsample Design & $n$ & $\hat{\bm{\gamma}}_1(0)^\top \hat{\bm{\gamma}}_1(1)$ & $T_n/\hat{\kappa}$ & $\tilde{T}_n$ & $P$-value ($T_n$) & $P$-value ($\tilde{T}_n$) \\
\hline
Full data & 48 & 0.9906 & 11.1155 & 12.2423 & 0.0021 & 0.0022 \\
\hline
Even IDs & 24 & 0.9969 & \hphantom{0}0.2755 & \hphantom{0}1.5889 & 0.3695 & 0.4518 \\
Odd IDs  & 24 & 0.9815 & 18.1055 & 15.8781 & $<0.0001$ & 0.0004 \\
\hline
ID $\equiv$ 0 (mod 3) & 16 & 0.9768 & 24.6772 & 25.7305 & $<0.0001$ & $<0.0001$ \\
ID $\equiv$ 1 (mod 3) & 16 & 0.9979 & \hphantom{0}0.2497 & \hphantom{0}1.3543 & 0.6834 & 0.5081 \\
ID $\equiv$ 2 (mod 3)  & 16 & 0.9896 & \hphantom{0}8.1640 & \hphantom{0}7.6338 & 0.0158 & 0.0220 \\
\hline
\end{tabular}
\end{small}
\end{table}

\begin{rem}
\cite{RefBFN99} also analyze the data with the response variables $y_1$ and $y_2$, and conclude that $\mathcal{H}_0$ is not rejected.
Under the same setting, our test statistics are
\[
T_n = 1.4251\hat{\kappa}, \quad
\tilde{T}_n = 1.4252 ,
\]
with the corresponding approximate $P$-values of $0.2326$ and $0.2327$, respectively, where $\hat{\kappa} = 0.0558$. 
The slight difference between $T_n/\hat{\kappa}$ and $\tilde{T}_n$ is due to numerical error.
\end{rem}

\section{Asymptotic normality in weak identifiability regimes}\label{sec:7}

Although the largest eigenvalue of $\bm{\Sigma}$ typically exhibits a dominant contribution ratio in allometry, we investigate the asymptotic distribution of the sample first principal eigenvector under a weak identifiability regimes~\citep{RefPRV20}.
Such a situation may arise in allometric regression when explanatory variables account for a substantial portion of the total variation.
Specifically, we consider the following setting, which is more restricted than that of \cite{RefTM25}.

\begin{ass}\label{ass2}
As $n\to\infty$ with $p,q,\bm{\alpha}$ fixed,
\begin{enumerate}
\item 
$\bm{X}_n \bm{\alpha} \neq \bm{0}$, and $n^{-1} \bm{X}_n^\top \bm{X}_n$ converges to a positive-definite matrix;
\item 
$\bm{\Sigma}$ depends on $n$ such that $\lambda_1 \to \lambda_{\infty} \in (0,\infty)$ and  $\lambda_1 - \lambda_p \asymp n^{-\xi}$ with some $\xi \in (0,\infty)$.
\end{enumerate}
\end{ass}

\begin{rem}
Unlike \cite{RefTM25}, where $\bm{X}_n \bm{\alpha} \neq \bm{0}_n$ is not necessarily assumed, we require this condition to show the asymptotic normality in a unified manner for all $w \in [0,1)$.
Moreover, to characterize the limiting distribution, the asymptotic behaviors of
$\lambda_1,\ldots,\lambda_p$ need to be specified.
Under Assumption~\ref{ass2}, it follows that $\lambda_k \to \lambda_{\infty}$ for all $k = 1 , \ldots, p$.
\end{rem}

Even when $\lambda_1- \lambda_2 \to 0$, the asymptotic normality presented in the following theorem holds for $w \in [0, 1)$ by an argument closely paralleling that of Theorem~\ref{thm1}, because the effective eigenvalue gap satisfies $\lambda_1(w) - \lambda_2(w) \sim n c_\infty \asymp n$.
This highlights that the hybrid PCA framework with $w<1$ yields asymptotic normality even when the leading principal direction is not asymptotically identifiable.
The case of $w=1$ is essentially identical to the ordinary PCA for an i.i.d.~sample, where $\lambda_1 - \lambda_2 \asymp n^{- \eta}$ with $\eta < 1/2$ is necessary and sufficient for the asymptotic normality; see Lemma 2.3 of \citet{RefPRV20}. 
We thus omit the details here.

\begin{thm}\label{thm2}
Let Assumption~\ref{ass2} hold.
For $w \in [0,1)$, as $n\to\infty$, 
\[
n^{1/2} (\hat{\bm{\gamma}}_1(w) - \bm{\gamma}_1)
\Rightarrow
\mathrm{N}_p \left( \bm{0}_p ,  \left\{ \left(\frac{w}{1-w} \frac{\lambda_{\infty}}{c_\infty} \right)^2 + \frac{\lambda_{\infty}}{c_\infty} \right\} (\bm{I}_p - \bm{\gamma}_1 \bm{\gamma}_1^\top) \right) .
\]
\end{thm}

\begin{rem}
The asymptotic covariance matrix in Theorem~\ref{thm2} is isotropic within the subspace orthogonal to $\bm{\gamma}_1$.
Its variance is governed by the ratio $\lambda_{\infty}/c_\infty$, which quantifies the relative strength of principal component variability to regression signal.
As $w \to 1$, the variance inflates due to the loss of the stabilizing effect of the regression component.
Under Assumption~\ref{ass2}, the asymptotically optimal choice is $w=0$.
A small $c_\infty$ also inflates the variance, reflecting the inherent difficulty of estimation when the regression information is weak.
\end{rem}

\section{Concluding remarks}\label{sec:8}
In this paper, we establish the asymptotic normality of the hybrid estimator of the leading principal eigenvector within the multivariate allometric regression model.
We further propose a consistent test of parallelism between the regression direction and the first principal eigenvector in a rank-one reduced-rank regression framework.
Numerical experiments demonstrate that the proposed test exhibits stable finite-sample performance, particularly in terms of size control under small gaps among the minor eigenvalues compared with the Wald type test and the method of \citet{RefBFN99}.
Our data analysis shows that the results obtained by the hybrid PCA are consistent with the statistical significance reported by \citet{RefBFN99} for the painted turtle carapace data. 

Several aspects remain to be further investigated.
First, since our proposed test statistic is formulated in terms of $\hat{\bm{\gamma}}_1(1)$, developing an alternative testing procedure that does not rely on $\hat{\bm{\gamma}}_1(1)$ is an important direction for future research, particularly in settings where identifiability of $\bm{\gamma}_1$ is weak.
Another challenging issue is to relax the normality assumption on $\bm{e}_1,\ldots,\bm{e}_n$ and extend our results to more general settings that allow for heavier-tailed distributions.
Additionally, incorporating heteroscedasticity into the model is an indispensable extension, as motivated by the structural variations observed in the turtle carapace data.
Finally, while the present paper focuses on the asymptotic normality of $\hat{\bm{\gamma}}_1(w), w \in [0,1]$, and hypothesis testing of parallelism, it would be of interest to develop a data-driven procedure for selecting $w$ based on the limiting MSE to improve the estimation of the leading principal eigenvector $\bm{\gamma}_1$.

\section*{Acknowledgements}
This work was supported by Japan Society for the Promotion of Science KAKENHI Grant Numbers 25K07133 (KT) and 26K14744 (SM).

\bibliographystyle{apalike}
\bibliography{ref-0}

\clearpage
\appendix
\section*{Appendix}
\section{Proofs}\label{sec:A1}

\subsection{Proofs for Section~\ref{sec:3}}

\begin{proof}[Proof of Theorem~\ref{thm1}]
Under Assumption~\ref{ass1}, $a$, $b$, and $d$ are constants, and $c = O(n)$.
Let $w \in [0,1]$.
It holds that $\lambda_1(w) \asymp n$ and that $\lambda_1(w) - \lambda_2(w) \asymp n$, which yield $\lambda_1(w) - \lambda_k(w) \asymp n$ for $k \in \{2,\ldots,p\}$.
From Lemma~\ref{lem1} and the fact that
\[ \E[ |\hat{\lambda}_1(w) - \lambda_1(w) |^2 ] \leq \E[\| \bm{S}(w) - \E[\bm{S}(w)]  \|_{\mathrm{F}}^2] = O(n), \]
it follows that
\[
\E[|\hat{\lambda}_1(w) - \lambda_1(w) |] \leq (\E[|\hat{\lambda}_1(w) - \lambda_1(w) |^2])^{1/2} = O(n^{1/2}), \quad
\E[ \| \hat{\bm{\gamma}}_1(w) - \bm{\gamma}_1 \|] = O(n^{-1/2}).
\]
These expectations yield 
\[
\bm{S}(w) - \E[\bm{S}(w)] = O_P(n^{1/2}), \quad
\hat{\lambda}_1(w) - \lambda_1(w) = O_P(n^{1/2}), \quad 
\hat{\bm{\gamma}}_1(w) - \bm{\gamma}_1 = O_P(n^{-1/2}).
\]

Since $\| \hat{\bm{\gamma}}_1(w) \|^2 = 1$, we have
$2 \bm{\gamma}_1^\top (\hat{\bm{\gamma}}_1(w) - \bm{\gamma}_1) + \| \hat{\bm{\gamma}}_1(w) - \bm{\gamma}_1 \|^2 = 0$, which implies
$\bm{\gamma}_1^\top (\hat{\bm{\gamma}}_1(w) - \bm{\gamma}_1) = O_P(n^{-1})$.
Let $a_k = \bm{\gamma}_k^\top (\hat{\bm{\gamma}}_1(w) - \bm{\gamma}_1)$ for $k \in \{1, \ldots, p\}$.
Then, it holds that
\[
\hat{\bm{\gamma}}_1(w) - \bm{\gamma}_1
= \sum_{k=1}^p a_k \bm{\gamma}_k
= \sum_{k=2}^p a_k \bm{\gamma}_k + O_P(n^{-1}).
\]
Using the characteristic equations
\[ \bm{S}(w) \hat{\bm{\gamma}}_1(w) = \hat{\lambda}_1 (w)\hat{\bm{\gamma}}_1(w), \quad 
\E[\bm{S}(w)] \bm{\gamma}_k = \lambda_k(w) \bm{\gamma}_k \quad (k \in \{ 1,\ldots,p\} ), \]
we have
\begin{align*}
& \{\E[\bm{S}(w)] + (\bm{S}(w) - \E[\bm{S}(w)])\}\{\bm{\gamma}_1 + (\hat{\bm{\gamma}}_1(w) - \bm{\gamma}_1) \} \\
& \qquad = \{  \lambda_1(w) + ( \hat{\lambda}_1 (w) - \lambda_1(w) )\} \{ \bm{\gamma}_1(w) + (\hat{\bm{\gamma}}_1(w) - \bm{\gamma}_1(w) )\}\\
&\Leftrightarrow
(\bm{S}(w) - \E[\bm{S}(w)])\bm{\gamma}_1 
+ \E[\bm{S}(w)] (\hat{\bm{\gamma}}_1(w) - \bm{\gamma}_1)
+ (\bm{S}(w) - \E[\bm{S}(w)])(\hat{\bm{\gamma}}_1(w) - \bm{\gamma}_1) \\
&\qquad
= ( \hat{\lambda}_1 (w) - \lambda_1(w) ) \bm{\gamma}_1
+ \lambda_1(w) (\hat{\bm{\gamma}}_1(w) - \bm{\gamma}_1 )
+ ( \hat{\lambda}_1 (w) - \lambda_1(w) )  (\hat{\bm{\gamma}}_1(w) - \bm{\gamma}_1 ) \\
&\Leftrightarrow 
\{(\bm{S}(w) - \E[\bm{S}(w)]) - ( \hat{\lambda}_1 (w) - \lambda_1(w) ) \bm{I}_p \} \bm{\gamma}_1 
= (\lambda_1(w) \bm{I}_p - \E[\bm{S}(w)]) (\hat{\bm{\gamma}}_1(w) - \bm{\gamma}_1 ) + \bm{r}_n,
\end{align*}
where the remainder term $\bm{r}_n = \{ ( \hat{\lambda}_1 (w) - \lambda_1(w) )\bm{I}_p - (\bm{S}(w) - \E[\bm{S}(w)])\} (\hat{\bm{\gamma}}_1(w) - \bm{\gamma}_1 )$ is $O_P(1)$.
For $k \in \{2,\ldots,p\}$, pre-multiplying the above equation by $\bm{\gamma}_k^\top$ yields
\begin{align*}
& \bm{\gamma}_k^\top (\bm{S}(w) - \E[\bm{S}(w)]) \bm{\gamma}_1 
= \bm{\gamma}_k^\top (\lambda_1(w) \bm{I}_p - \E[\bm{S}(w)]) \left( \sum_{j = 1}^p a_j \bm{\gamma}_j \right) + \bm{\gamma}_k^\top \bm{r}_n \\
&= (\lambda_1(w) - \lambda_k(w) ) \bm{\gamma}_k^\top \left( \sum_{j=1}^p a_j \bm{\gamma}_j \right) + \bm{\gamma}_k^\top \bm{r}_n 
= a_k (\lambda_1(w) - \lambda_k(w) ) + \bm{\gamma}_k^\top \bm{r}_n ,
\end{align*}
from which we have
\[
a_k = \frac{ \bm{\gamma}_k^\top (\bm{S}(w) - \E[\bm{S}(w)]) \bm{\gamma}_1  - \bm{\gamma}_k^\top \bm{r}_n }{\lambda_1(w) - \lambda_k(w) }
=  \frac{ \bm{\gamma}_k^\top (\bm{S}(w) - \E[\bm{S}(w)]) \bm{\gamma}_1  }{\lambda_1(w) - \lambda_k(w) } + O_P(n^{-1}).
\]
We thus have
\[
n^{1/2} (\hat{\bm{\gamma}}_1(w) - \bm{\gamma}_1)
= 
\sum_{k=2}^p \frac{ \bm{\gamma}_k^\top \{ n^{-1/2} (\bm{S}(w) - \E[\bm{S}(w)]) \} \bm{\gamma}_1 }{n^{-1} (\lambda_1(w) - \lambda_k(w))}  \bm{\gamma}_k + O_P(n^{-1/2}).
\]

Let $\tilde{\bm{\Lambda}} := \diag(\lambda_2,\ldots,\lambda_p)$ and
\[
\bm{g}_n 
:= (\bm{\gamma}_2^\top \{ n^{-1/2} (\bm{S}(w) - \E[\bm{S}(w)]) \} \bm{\gamma}_1, \ldots, \bm{\gamma}_p^\top \{ n^{-1/2} (\bm{S}(w) - \E[\bm{S}(w)]) \} \bm{\gamma}_1  )^\top .
\]
Using the Cram\'{e}r--Wold device, we shall show
\begin{equation}
    \bm{g}_n 
\Rightarrow
\mathrm{N}_{p-1} \left( \bm{0}, \{(1-w)^2 c_\infty + w^2 \lambda_1 \} \tilde{\bm{\Lambda}}  \right). \label{thp0}    
\end{equation}
Fix an arbitrary $\bm{t} = (t_2,\ldots,t_p)^\top \in \R^{p-1}$.
We examine the convergence in distribution of
\[
\bm{t}^\top \bm{g}_n =
\sum_{k=2}^p t_k \bm{\gamma}_k^\top \{ n^{-1/2} (\bm{S}(w) - \E[\bm{S}(w)]) \} \bm{\gamma}_1 .
\]
Since
\[ \bm{e}_1,\ldots,\bm{e}_n \stackrel{\mathrm{i.i.d.}}{\sim} \mathrm{N}_p (\bm{0}_p, \bm{\Sigma}),\]
the two terms on the right-hand side of
\[
n^{-1/2} (\bm{S}(w) - \E[\bm{S}(w)]) 
= (1-w) n^{1/2} (\bm{S}_R/n - \E[\bm{S}_R/n]) + w n^{1/2}(\bm{S}_E/n - \E[\bm{S}_E/n]) 
\]
are independent, so we consider the convergences in distribution of
\begin{align}
& (1-w) \sum_{k=2}^p t_k \bm{\gamma}_k^\top n^{1/2} (\bm{S}_R/n - \E[\bm{S}_R/n]) \bm{\gamma}_1 , \label{thp2} \\
&w \sum_{k=2}^p t_k \bm{\gamma}_k^\top n^{1/2}(\bm{S}_E/n - \E[\bm{S}_E/n])  \bm{\gamma}_1, \label{thp3}
\end{align}
separately.
Consider an i.i.d.~random sequence $\bm{z}_1,\bm{z}_2,\ldots \sim \mathrm{N}_p (\bm{0}_p , \bm{\Sigma})$.
Hereafter, $\stackrel{\mathrm{d}}{=}$ denotes the equality in distribution.
\begin{itemize}
\item 
Consider \eqref{thp2}.
From a stochastic representation
\[
\bm{S}_R
\stackrel{\mathrm{d}}{=} \|\bm{X}_n \bm{\alpha} \|^2 \bm{\gamma}_1 \bm{\gamma}_1^\top 
+  \|\bm{X}_n \bm{\alpha} \| ( \bm{\gamma}_1 \bm{z}_1^\top +   \bm{z}_1 \bm{\gamma}_1^\top ) 
+ \sum_{i=1}^q \bm{z}_i \bm{z}_i^\top ,
\]
it follows that
\[
n^{1/2} (\bm{S}_R/n - \E[\bm{S}_R/n]) 
\stackrel{\mathrm{d}}{=} n^{-1/2} \|\bm{X}_n \bm{\alpha} \| ( \bm{\gamma}_1 \bm{z}_1^\top +  \bm{z}_1 \bm{\gamma}_1^\top ) + n^{-1/2} \sum_{i=1}^q (\bm{z}_i \bm{z}_i^\top - \bm{\Sigma} ).
\]
Since $q$ is a constant, the second term on the right-hand side is $O_P(n^{-1/2})$.
For $k \neq 1$, we have
\[
\bm{\gamma}_k^\top 
\{ n^{-1/2} \|\bm{X}_n \bm{\alpha} \| ( \bm{\gamma}_1   \bm{z}_1^\top + \bm{z}_1 \bm{\gamma}_1^\top ) \} \bm{\gamma}_1
= n^{-1/2} \|\bm{X}_n \bm{\alpha} \| \bm{\gamma}_k^\top \bm{z}_1 . \]
Consequently, as $n\to\infty$, we obtain
\begin{align*}
& \sum_{k=2}^p t_k \bm{\gamma}_k^\top n^{1/2} (\bm{S}_R/n - \E[\bm{S}_R/n]) \bm{\gamma}_1 \\
& \stackrel{\mathrm{d}}{=} \sum_{k=2}^p t_k n^{-1/2} \|\bm{X}_n \bm{\alpha} \|  \bm{\gamma}_k^\top \bm{z}_1  + o_P(1)
\Rightarrow     
\mathrm{N} \left(0, c_\infty \sum_{k=2}^p t_k^2 \lambda_k \right).
\end{align*}
\item
Consider \eqref{thp3}.
From a stochastic representation
\[
\bm{S}_E
\stackrel{\mathrm{d}}{=} \sum_{i=1}^{n-q-1} \bm{z}_i \bm{z}_i^\top ,
\]
it follows that
\[
n^{1/2} (\bm{S}_E/n - \E[\bm{S}_E/n]) 
\stackrel{\mathrm{d}}{=} n^{-1/2} \sum_{i=1}^{n-q-1} (\bm{z}_i \bm{z}_i^\top - \bm{\Sigma} ) .
\]
Since
\[
\bm{\gamma}_k^\top
\left\{n^{-1/2} \sum_{i=1}^{n-q-1} (\bm{z}_i \bm{z}_i^\top - \bm{\Sigma} )  \right\} \bm{\gamma}_1
= n^{-1/2} \sum_{i=1}^{n-q-1}  \bm{\gamma}_k^\top\bm{z}_i \bm{z}_i^\top \bm{\gamma}_1 , \]
the central limit theorem yields
\begin{align*}
& \sum_{k=2}^p t_k \bm{\gamma}_k^\top n^{1/2} (\bm{S}_E/n - \E[\bm{S}_E/n]) \bm{\gamma}_1  \\
& \stackrel{\mathrm{d}}{=} \sum_{k=2}^p t_k n^{-1/2} \sum_{i=1}^{n-q-1}  \bm{\gamma}_k^\top\bm{z}_i \bm{z}_i^\top \bm{\gamma}_1
\Rightarrow     
\mathrm{N} \left(0, \lambda_1 \sum_{k=2}^p t_k^2 \lambda_k  \right).
\end{align*}
\end{itemize}
The independence of \eqref{thp2} and \eqref{thp3} implies that
\[
\bm{t}^\top \bm{g}_n =
\sum_{k=2}^p t_k \bm{\gamma}_k^\top \{ n^{-1/2} (\bm{S}(w) - \E[\bm{S}(w)]) \} \bm{\gamma}_1 
\Rightarrow
\mathrm{N} \left(0, \{(1-w)^2  c_\infty + w^2 \lambda_1 \} \bm{t}^\top \tilde{\bm{\Lambda}} \bm{t}  \right).
\]
Hence, the Cram\'{e}r--Wold device yields \eqref{thp0}.
In particular, the limiting distribution is a $(p-1)$-dimensional normal distribution with a diagonal covariance matrix.
Consequently, it follows from
\[
n^{-1} (\lambda_1(w) - \lambda_k(w)) \to w (\lambda_1 - \lambda_k) + (1-w) c_\infty
\]
that
\begin{align*}
& n^{1/2} (\hat{\bm{\gamma}}_1(w) - \bm{\gamma}_1)
= 
\sum_{k=2}^p \frac{ \bm{\gamma}_k^\top \{ n^{-1/2} (\bm{S}(w) - \E[\bm{S}(w)]) \} \bm{\gamma}_1 }{n^{-1} (\lambda_1(w) - \lambda_k(w))}  \bm{\gamma}_k + O_P(n^{-1/2}) \\
&\Rightarrow
\mathrm{N}_p \left( \bm{0}_p, \{w^2 \lambda_1 + (1-w)^2 c_\infty  \} \sum_{k=2}^p \frac{  \lambda_k}{\{ w (\lambda_1 - \lambda_k) + (1-w) c_\infty \}^2} \bm{\gamma}_k \bm{\gamma}_k^\top \right).
\end{align*}
This completes the proof.
\end{proof}

\begin{proof}[Proof of Corollary~\ref{cor0}]
By the Davis--Kahn theorem and moment bounds for $\| \bm{S}(w) - \E[\bm{S}(w)] \|_{\mathrm{F}}$ under normality, the sequence $\{ \| n^{1/2} (\hat{\bm{\gamma}}_1(w) - \bm{\gamma}_1) \|^2 \}_{n \in \N}$ is uniformly integrable.
Therefore, Theorem~\ref{thm1} yields the conclusion.
\end{proof}

\begin{proof}[Proof of Corollary~\ref{cor1}]
Since $\bm{S}_R$ and $\bm{S}_E$ are independent, so are $\hat{\bm{\gamma}}_1(0)$ and $\hat{\bm{\gamma}}_1(1)$.
The conclusion then follows from Theorem~\ref{thm1}.
\end{proof}

\begin{proof}[Proof of Corollary~\ref{cor2}]
From Corollary~\ref{cor1}, the continuous mapping theorem yields
\[
2n(1 - \hat{\bm{\gamma}}_1(1)^\top \hat{\bm{\gamma}}_1(0))
= n \|\hat{\bm{\gamma}}_1(1) - \hat{\bm{\gamma}}_1(0) \|^2 
\Rightarrow \sum_{k=2}^p \left\{ \frac{  \lambda_1}{ (\lambda_1 - \lambda_k)^2}  + \frac{1}{c_\infty} \right\} \lambda_k  \chi^2_{1,k}.
\]
By the Slutsky theorem,
\begin{align*}
& n\{ 1 - (\hat{\bm{\gamma}}_1(1)^\top \hat{\bm{\gamma}}_1(0))^2 \} - 2n(1 - \hat{\bm{\gamma}}_1(1)^\top \hat{\bm{\gamma}}_1(0)) \\
&= (1 + \hat{\bm{\gamma}}_1(1)^\top \hat{\bm{\gamma}}_1(0)) n(1 - \hat{\bm{\gamma}}_1(1)^\top \hat{\bm{\gamma}}_1(0)) - 2n(1 - \hat{\bm{\gamma}}_1(1)^\top \hat{\bm{\gamma}}_1(0)) \\
&= (\hat{\bm{\gamma}}_1(1)^\top \hat{\bm{\gamma}}_1(0) - 1) n(1 - \hat{\bm{\gamma}}_1(1)^\top \hat{\bm{\gamma}}_1(0))
= O_P(n^{-1}) ,
\end{align*}
which yields the conclusion.
\end{proof}

\subsection{Proofs for Section~\ref{sec:4}}

\begin{proof}[Proof of Lemma~\ref{lem2}]
It follows from $\hat{\bm{\Sigma}} \stackrel{\mathrm{p}}{\to} \bm{\Sigma}$ that
\[ \max_{k \in \{1,\ldots,p \}} | \hat{\lambda}_k - \lambda_k | \leq \| \hat{\bm{\Sigma}} - \bm{\Sigma} \|_{\mathrm{F}} \stackrel{\mathrm{p}}{\to} 0 . \]
By the same argument as in Proof of Theorem~\ref{thm1}, 
\[ n^{-1} \bm{S}_R - \E[n^{-1} \bm{S}_R]  = O_P(n^{-1/2}). \]
Since 
\[ \E[n^{-1} \tr(\bm{S}_R)] = n^{-1} c + n^{-1} q \tr(\bm{\Sigma}) \to c_{\infty}, \]
it holds that
\[ n^{-1} \tr(\bm{S}_R) \stackrel{\mathrm{p}}{\to} c_{\infty}. \]
Combined with 
\[ n^{-1} q \tr(\hat{\bm{\Sigma}}) \stackrel{\mathrm{p}}{\to} 0,\]
we have
\[ \hat{c}_\infty = n^{-1} (\tr(\bm{S}_R)  - q \tr(\hat{\bm{\Sigma}})) \stackrel{\mathrm{p}}{\to} c_\infty.\]
The continuous mapping theorem then yields
\[
\hat{V}_k \stackrel{\mathrm{p}}{\to}
\left\{ \frac{  \lambda_1}{ (\lambda_1 - \lambda_k)^2}  + \frac{1}{c_\infty} \right\} \lambda_k 
\]
for $k \in \{2,\ldots,p\}$.
Consequently,
\begin{align*}
& \left| \sum_{k=2}^p \left[ \hat{V}_k -
\left\{ \frac{  \lambda_1}{ (\lambda_1 - \lambda_k)^2}  + \frac{1}{c_\infty} \right\} \lambda_k \right] \chi^2_{1,k} \right| \\
& \leq 
\max_{k\in\{2,\ldots,p\}} \left| \hat{V}_k -
\left\{ \frac{  \lambda_1}{ (\lambda_1 - \lambda_k)^2}  + \frac{1}{c_\infty} \right\} \lambda_k \right| \sum_{k=2}^p  \chi^2_{1,k} 
\stackrel{\mathrm{p}}{\to} 0.
\end{align*}
This completes the proof.
\end{proof}

\begin{proof}[Proof of Theorem~\ref{thm3}]
Using Lemma~\ref{lem2} and the fact that
\[
\sum_{k=2}^p \left\{ \frac{  \lambda_1}{ (\lambda_1 - \lambda_k)^2}  + \frac{1}{c_\infty} \right\} \lambda_k  \chi^2_{1,k}
\]
has a continuous distribution, the continuity of the quantile mapping yields $\hat{q}(\alpha_0) \stackrel{\mathrm{p}}{\to} q(\alpha_0)$.
\begin{itemize}
\item[(i)]
From Corollary~\ref{cor2}, the assertion follows.
\item[(ii)]
Theorem~\ref{thm1} yields $n^{1/2} ({\bm{\gamma}}_1 - \hat{\bm{\gamma}}_1(1)) = O_P(1)$.
Moreover, $\hat{\bm{\gamma}}_1(1)=O_P(1)$.
By a similar argument to the proof of Lemma~\ref{lem1}, we have $n^{1/2} ( \hat{\bm{\gamma}}_1(0) - {\bm{\beta}}) = O_P(1)$.
Consequently,
\begin{align*}
T_n 
&= n \|\hat{\bm{\gamma}}_1(1) - \hat{\bm{\gamma}}_1(0) \|^2 
= 2 n (1 - \hat{\bm{\gamma}}_1(1)^\top \hat{\bm{\gamma}}_1(0)) \\
&= 2 n (1 - {\bm{\gamma}}_1^\top {\bm{\beta}}) + 2 n ({\bm{\gamma}}_1^\top {\bm{\beta}} - \hat{\bm{\gamma}}_1(1)^\top \hat{\bm{\gamma}}_1(0)) \\
&= 2 n (1 - {\bm{\gamma}}_1^\top {\bm{\beta}}) 
+ 2 n ({\bm{\gamma}}_1 - \hat{\bm{\gamma}}_1(1))^\top  {\bm{\beta}} 
+ 2 \hat{\bm{\gamma}}_1(1)^\top n( {\bm{\beta}} -  \hat{\bm{\gamma}}_1(0)) \\
&= 2 n (1 - {\bm{\gamma}}_1^\top {\bm{\beta}}) + O_P(n^{1/2}).
\end{align*}
Since $|\bm{\beta}^\top \bm{\gamma}_1| < 1$, $1 - {\bm{\gamma}}_1^\top {\bm{\beta}}$ is a positive constant.
\end{itemize}
This completes the proof.
\end{proof}

\section{Detailed numerical results}\label{sec:A2}

Empirical type I and type II error rates for the binary and continuous explanatory variable cases are reported in Tables~\ref{tab:size1}--\ref{tab:size2} and Tables~\ref{tab:power1}--\ref{tab:power2}, respectively.
In particular, Tables~\ref{tab:size1}--\ref{tab:size2} provide detailed numerical results corresponding to the figures in the main text.

\begin{table}[htbp]
\centering
\small
\setlength{\tabcolsep}{2.8pt}
\renewcommand{\arraystretch}{1.1}
\caption{Empirical type I error rates for the binary cases with eigenvalues $(\lambda_1,\ldots,\lambda_5)$ and $c/n = n^{-1} \| \bm{X}_n \bm{\alpha} \|^2$ for $n \in \{20, 50, 100, 200, 500\}$.
Bold entries indicate rates not exceeding the significance level $\alpha_0 = 0.05$.}
\label{tab:size1}

\begin{tabular}{llc ccccc}
\toprule
$(\lambda_1,\ldots,\lambda_5)$ & $c/n$ & Method & 20 & 50 & 100 & 200 & 500 \\
\midrule

$(10,8,6,4,2)$ & 2 
& $\varphi_n$ & 0.062 & 0.076 & 0.092 & 0.072 & 0.053 \\
& & $\tilde{\varphi}_n$ & 0.114 & 0.061 & 0.068 & 0.056 & 0.067 \\
& & $\varphi_{\mathrm{BFN}}$ & 0.641 & 0.500 & 0.435 & 0.315 & 0.168 \\
\cline{2-8}
& 10
& $\varphi_n$ & 0.177 & 0.163 & 0.127 & 0.088 & 0.060 \\
& & $\tilde{\varphi}_n$ & 0.332 & 0.176 & 0.114 & 0.090 & 0.087 \\
& & $\varphi_{\mathrm{BFN}}$ & 0.838 & 0.707 & 0.533 & 0.404 & 0.230 \\
\midrule

$(10,1,1,1,1)$ & 2
& $\varphi_n$ & \textbf{0.021} & \textbf{0.026} & \textbf{0.048} & \textbf{0.046} & \textbf{0.045} \\
& & $\tilde{\varphi}_n$ & 0.083 & \textbf{0.040} & 0.058 & \textbf{0.050} & \textbf{0.047} \\
& & $\varphi_{\mathrm{BFN}}$ & 0.372 & 0.097 & 0.086 & \textbf{0.050} & 0.052 \\
\cline{2-8}
& 10
& $\varphi_n$ & 0.067 & 0.051 & 0.055 & 0.056 & \textbf{0.049} \\
& & $\tilde{\varphi}_n$ & 0.168 & 0.109 & 0.067 & 0.061 & \textbf{0.049} \\
& & $\varphi_{\mathrm{BFN}}$ & 0.431 & 0.152 & 0.092 & 0.068 & \textbf{0.047} \\
\midrule

$(2,1.5,1.5,1.5,1)$ & 2
& $\varphi_n$ & 0.147 & 0.134 & 0.101 & 0.077 & \textbf{0.043} \\
& & $\tilde{\varphi}_n$ & 0.357 & 0.242 & 0.173 & 0.112 & 0.074 \\
& & $\varphi_{\mathrm{BFN}}$ & 0.904 & 0.843 & 0.746 & 0.563 & 0.286 \\
\cline{2-8}
& 10
& $\varphi_n$ & 0.214 & 0.161 & 0.120 & \textbf{0.050} & \textbf{0.048} \\
& & $\tilde{\varphi}_n$ & 0.555 & 0.340 & 0.233 & 0.153 & 0.087 \\
& & $\varphi_{\mathrm{BFN}}$ & 0.968 & 0.911 & 0.790 & 0.615 & 0.305 \\
\midrule

$(1.8,1,1,1,1)$ & 2
& $\varphi_n$ & 0.105 & 0.059 & 0.054 & \textbf{0.047} & \textbf{0.046} \\
& & $\tilde{\varphi}_n$ & 0.316 & 0.141 & 0.103 & 0.071 & 0.053 \\
& & $\varphi_{\mathrm{BFN}}$ & 0.853 & 0.657 & 0.438 & 0.214 & 0.115 \\
\cline{2-8}
& 10
& $\varphi_n$ & 0.143 & 0.066 & 0.061 & 0.060 & \textbf{0.049} \\
& & $\tilde{\varphi}_n$ & 0.422 & 0.201 & 0.118 & 0.106 & 0.069 \\
& & $\varphi_{\mathrm{BFN}}$ & 0.927 & 0.752 & 0.467 & 0.277 & 0.116 \\
\midrule

$(1.2,1,1,1,1)$ & 2
& $\varphi_n$ & 0.195 & 0.200 & 0.154 & 0.113 & 0.090 \\
& & $\tilde{\varphi}_n$ & 0.510 & 0.444 & 0.345 & 0.236 & 0.137 \\
& & $\varphi_{\mathrm{BFN}}$ & 0.960 & 0.939 & 0.933 & 0.870 & 0.692 \\
\cline{2-8}
& 10
& $\varphi_n$ & 0.254 & 0.212 & 0.168 & 0.141 & 0.091 \\
& & $\tilde{\varphi}_n$ & 0.644 & 0.487 & 0.396 & 0.304 & 0.178 \\
& & $\varphi_{\mathrm{BFN}}$ & 0.992 & 0.964 & 0.944 & 0.867 & 0.714 \\
\bottomrule
\end{tabular}
\end{table}

\begin{table}[htbp]
\centering
\small
\setlength{\tabcolsep}{2.8pt}
\renewcommand{\arraystretch}{1.1}
\caption{Empirical type I error rates for the continuous cases with eigenvalues $(\lambda_1,\ldots,\lambda_5)$ and $c/n = n^{-1} \| \bm{X}_n \bm{\alpha} \|^2$ for $n \in \{20, 50, 100, 200, 500\}$.
Bold entries indicate rates not exceeding the significance level $\alpha_0 = 0.05$.}
\label{tab:size2}

\begin{tabular}{llc ccccc}
\toprule
$(\lambda_1,\ldots,\lambda_5)$ & $c/n$ & Method & 20 & 50 & 100 & 200 & 500 \\
\midrule

$(10,8,6,4,2)$ & 2
& $\varphi_n$ & 0.060 & 0.075 & 0.111 & 0.064 & 0.057 \\
& & $\tilde{\varphi}_n$ & 0.117 & 0.052 & 0.063 & \textbf{0.049} & 0.052 \\
\cmidrule(l){2-8}
& 10
& $\varphi_n$ & 0.167 & 0.166 & 0.127 & 0.089 & 0.073 \\
& & $\tilde{\varphi}_n$ & 0.297 & 0.191 & 0.119 & 0.090 & 0.079 \\
\midrule

$(10,1,1,1,1)$ & 2
& $\varphi_n$ & \textbf{0.023} & \textbf{0.024} & \textbf{0.032} & \textbf{0.042} & \textbf{0.045} \\
& & $\tilde{\varphi}_n$ & 0.095 & \textbf{0.036} & \textbf{0.046} & 0.052 & \textbf{0.047} \\
\cmidrule(l){2-8}
& 10
& $\varphi_n$ & 0.057 & \textbf{0.045} & 0.055 & 0.059 & 0.052 \\
& & $\tilde{\varphi}_n$ & 0.162 & 0.089 & 0.069 & 0.061 & 0.052 \\
\midrule

$(2,1.5,1.5,1.5,1)$ & 2
& $\varphi_n$ & 0.139 & 0.142 & 0.111 & 0.073 & \textbf{0.044} \\
& & $\tilde{\varphi}_n$ & 0.332 & 0.249 & 0.172 & 0.106 & 0.063 \\
\cmidrule(l){2-8}
& 10
& $\varphi_n$ & 0.202 & 0.143 & 0.115 & 0.057 & \textbf{0.044} \\
& & $\tilde{\varphi}_n$ & 0.549 & 0.341 & 0.229 & 0.147 & 0.086 \\
\midrule

$(1.8,1,1,1,1)$ & 2
& $\varphi_n$ & 0.093 & 0.065 & \textbf{0.041} & \textbf{0.047} & 0.052 \\
& & $\tilde{\varphi}_n$ & 0.291 & 0.139 & 0.105 & 0.085 & 0.059 \\
\cmidrule(l){2-8}
& 10
& $\varphi_n$ & 0.143 & 0.073 & 0.060 & 0.060 & \textbf{0.046} \\
& & $\tilde{\varphi}_n$ & 0.430 & 0.194 & 0.116 & 0.097 & 0.069 \\
\midrule

$(1.2,1,1,1,1)$ & 2
& $\varphi_n$ & 0.188 & 0.191 & 0.151 & 0.110 & 0.091 \\
& & $\tilde{\varphi}_n$ & 0.531 & 0.428 & 0.323 & 0.259 & 0.149 \\
\cmidrule(l){2-8}
& 10
& $\varphi_n$ & 0.253 & 0.200 & 0.164 & 0.131 & 0.091 \\
& & $\tilde{\varphi}_n$ & 0.656 & 0.482 & 0.387 & 0.296 & 0.176 \\
\bottomrule
\end{tabular}
\end{table}

\begin{table}[htbp]
\centering
\small
\setlength{\tabcolsep}{2.8pt}
\renewcommand{\arraystretch}{1.1}
\caption{Empirical rejection rates under $\Hc_1$ for the binary cases with $\bm{\beta} = (\cos \theta, \sin \theta, 0, \ldots, 0)^\top$, eigenvalues $(\lambda_1,\ldots,\lambda_5)$, and $c/n = n^{-1} \| \bm{X}_n \bm{\alpha} \|^2$ for $n \in \{20, 50, 100, 200, 500\}$, with significance level $\alpha_0 = 0.05$.}
\label{tab:power1}
\begin{tabular}{llc ccccc @{\hspace{8pt}} ccccc}
\toprule
& & 
& \multicolumn{5}{c}{$\theta = \pi/4$}
& \multicolumn{5}{c}{$\theta = \pi/2$} \\
\midrule
$(\lambda_1,\ldots,\lambda_5)$ & $c/n$ & Method
& 20 & 50 & 100 & 200 & 500
& 20 & 50 & 100 & 200 & 500 \\
\midrule

$(10,8,6,4,2)$ & 2
& $\varphi_n$ & 0.079 & 0.218 & 0.364 & 0.567 & 0.850 & 0.118 & 0.363 & 0.619 & 0.826 & 0.980 \\
& & $\tilde{\varphi}_n$ & 0.105 & 0.118 & 0.200 & 0.380 & 0.711 & 0.121 & 0.194 & 0.405 & 0.672 & 0.918 \\
& & $\varphi_{\mathrm{BFN}}$ & 0.652 & 0.561 & 0.637 & 0.722 & 0.858 & 0.686 & 0.592 & 0.728 & 0.874 & 0.988 \\
\cline{2-13}
& 10
& $\varphi_n$ & 0.269 & 0.370 & 0.481 & 0.680 & 0.870 & 0.362 & 0.565 & 0.688 & 0.859 & 0.984 \\
& & $\tilde{\varphi}_n$ & 0.404 & 0.375 & 0.433 & 0.552 & 0.770 & 0.489 & 0.530 & 0.580 & 0.738 & 0.933 \\
& & $\varphi_{\mathrm{BFN}}$ & 0.889 & 0.871 & 0.854 & 0.888 & 0.942 & 0.931 & 0.944 & 0.970 & 0.990 & 0.999 \\
\midrule

$(10,1,1,1,1)$ & 2
& $\varphi_n$ & 0.730 & 1.000 & 1.000 & 1.000 & 1.000 & 0.974 & 1.000 & 1.000 & 1.000 & 1.000 \\
& & $\tilde{\varphi}_n$ & 0.767 & 1.000 & 1.000 & 1.000 & 1.000 & 0.970 & 1.000 & 1.000 & 1.000 & 1.000 \\
& & $\varphi_{\mathrm{BFN}}$ & 0.915 & 1.000 & 1.000 & 1.000 & 1.000 & 0.994 & 1.000 & 1.000 & 1.000 & 1.000 \\
\cline{2-13}
& 10
& $\varphi_n$ & 0.996 & 1.000 & 1.000 & 1.000 & 1.000 & 1.000 & 1.000 & 1.000 & 1.000 & 1.000 \\
& & $\tilde{\varphi}_n$ & 1.000 & 1.000 & 1.000 & 1.000 & 1.000 & 1.000 & 1.000 & 1.000 & 1.000 & 1.000 \\
& & $\varphi_{\mathrm{BFN}}$ & 0.999 & 1.000 & 1.000 & 1.000 & 1.000 & 1.000 & 1.000 & 1.000 & 1.000 & 1.000 \\
\midrule

$(2,1.5,1.5,1.5,1)$ & 2
& $\varphi_n$ & 0.213 & 0.278 & 0.402 & 0.557 & 0.890 & 0.289 & 0.463 & 0.628 & 0.818 & 0.982 \\
& & $\tilde{\varphi}_n$ & 0.487 & 0.467 & 0.588 & 0.732 & 0.929 & 0.583 & 0.702 & 0.803 & 0.911 & 0.991 \\
& & $\varphi_{\mathrm{BFN}}$ & 0.940 & 0.941 & 0.946 & 0.956 & 0.984 & 0.963 & 0.981 & 0.995 & 0.998 & 1.000 \\
\cline{2-13}
& 10
& $\varphi_n$ & 0.286 & 0.312 & 0.420 & 0.566 & 0.901 & 0.378 & 0.486 & 0.652 & 0.806 & 0.984 \\
& & $\tilde{\varphi}_n$ & 0.675 & 0.588 & 0.650 & 0.733 & 0.933 & 0.777 & 0.761 & 0.841 & 0.917 & 0.992 \\
& & $\varphi_{\mathrm{BFN}}$ & 0.990 & 0.974 & 0.976 & 0.966 & 0.993 & 0.993 & 0.997 & 0.995 & 0.999 & 1.000 \\
\midrule

$(1.8,1,1,1,1)$ & 2
& $\varphi_n$ & 0.266 & 0.490 & 0.745 & 0.968 & 1.000 & 0.471 & 0.761 & 0.955 & 0.997 & 1.000 \\
& & $\tilde{\varphi}_n$ & 0.593 & 0.735 & 0.882 & 0.983 & 1.000 & 0.775 & 0.915 & 0.989 & 0.999 & 1.000 \\
& & $\varphi_{\mathrm{BFN}}$ & 0.959 & 0.955 & 0.983 & 0.999 & 1.000 & 0.991 & 1.000 & 1.000 & 1.000 & 1.000 \\
\cline{2-13}
& 10
& $\varphi_n$ & 0.354 & 0.530 & 0.776 & 0.982 & 1.000 & 0.507 & 0.762 & 0.950 & 0.999 & 1.000 \\
& & $\tilde{\varphi}_n$ & 0.749 & 0.790 & 0.916 & 0.994 & 1.000 & 0.880 & 0.943 & 0.978 & 0.999 & 1.000 \\
& & $\varphi_{\mathrm{BFN}}$ & 0.991 & 0.987 & 0.991 & 0.999 & 1.000 & 0.999 & 0.999 & 1.000 & 1.000 & 1.000 \\
\midrule

$(1.2,1,1,1,1)$ & 2
& $\varphi_n$ & 0.256 & 0.257 & 0.300 & 0.337 & 0.469 & 0.294 & 0.367 & 0.414 & 0.541 & 0.758 \\
& & $\tilde{\varphi}_n$ & 0.637 & 0.591 & 0.591 & 0.654 & 0.764 & 0.671 & 0.717 & 0.764 & 0.833 & 0.924 \\
& & $\varphi_{\mathrm{BFN}}$ & 0.982 & 0.980 & 0.979 & 0.977 & 0.982 & 0.992 & 0.996 & 0.996 & 0.996 & 0.999 \\
\cline{2-13}
& 10
& $\varphi_n$ & 0.291 & 0.286 & 0.278 & 0.338 & 0.510 & 0.351 & 0.361 & 0.446 & 0.543 & 0.755 \\
& & $\tilde{\varphi}_n$ & 0.731 & 0.647 & 0.635 & 0.639 & 0.777 & 0.813 & 0.758 & 0.787 & 0.851 & 0.937 \\
& & $\varphi_{\mathrm{BFN}}$ & 0.995 & 0.989 & 0.976 & 0.986 & 0.988 & 0.998 & 0.998 & 0.997 & 1.000 & 1.000 \\
\bottomrule
\end{tabular}
\end{table}

\begin{table}[htbp]
\centering
\small
\setlength{\tabcolsep}{2.8pt}
\renewcommand{\arraystretch}{1.1}
\caption{Empirical rejection rates under $\Hc_1$ for the continuous cases with $\bm{\beta} = (\cos \theta, \sin \theta, 0, \ldots, 0)^\top$, eigenvalues $(\lambda_1,\ldots,\lambda_5)$, and $c/n = n^{-1} \| \bm{X}_n \bm{\alpha} \|^2$ for $n \in \{20, 50, 100, 200, 500\}$, with significance level $\alpha_0 = 0.05$.}
\label{tab:power2}

\begin{tabular}{llc ccccc @{\hspace{8pt}} ccccc}
\toprule
& &
& \multicolumn{5}{c}{$\theta=\pi/4$}
& \multicolumn{5}{c}{$\theta=\pi/2$} \\
\midrule
$(\lambda_1,\ldots,\lambda_5)$ & $c/n$ & Method
& 20 & 50 & 100 & 200 & 500
& 20 & 50 & 100 & 200 & 500 \\
\midrule

$(10,8,6,4,2)$ & 2
& $\varphi_n$ & 0.083 & 0.233 & 0.377 & 0.521 & 0.865 & 0.137 & 0.356 & 0.628 & 0.831 & 0.975 \\
& & $\tilde{\varphi}_n$ & 0.122 & 0.141 & 0.208 & 0.360 & 0.706 & 0.130 & 0.176 & 0.392 & 0.631 & 0.905 \\
\cmidrule(l){2-13}
& 10
& $\varphi_n$ & 0.274 & 0.351 & 0.497 & 0.681 & 0.870 & 0.375 & 0.562 & 0.687 & 0.857 & 0.982 \\
& & $\tilde{\varphi}_n$ & 0.415 & 0.363 & 0.426 & 0.540 & 0.764 & 0.485 & 0.529 & 0.599 & 0.733 & 0.940 \\
\midrule

$(10,1,1,1,1)$ & 2
& $\varphi_n$ & 0.745 & 1.000 & 1.000 & 1.000 & 1.000 & 0.979 & 1.000 & 1.000 & 1.000 & 1.000 \\
& & $\tilde{\varphi}_n$ & 0.773 & 0.999 & 1.000 & 1.000 & 1.000 & 0.975 & 1.000 & 1.000 & 1.000 & 1.000 \\
\cmidrule(l){2-13}
& 10
& $\varphi_n$ & 0.996 & 1.000 & 1.000 & 1.000 & 1.000 & 1.000 & 1.000 & 1.000 & 1.000 & 1.000 \\
& & $\tilde{\varphi}_n$ & 1.000 & 1.000 & 1.000 & 1.000 & 1.000 & 1.000 & 1.000 & 1.000 & 1.000 & 1.000 \\
\midrule

$(2,1.5,1.5,1.5,1)$ & 2
& $\varphi_n$ & 0.211 & 0.266 & 0.396 & 0.570 & 0.898 & 0.285 & 0.474 & 0.623 & 0.823 & 0.981 \\
& & $\tilde{\varphi}_n$ & 0.485 & 0.483 & 0.583 & 0.724 & 0.936 & 0.565 & 0.706 & 0.819 & 0.909 & 0.991 \\
\cmidrule(l){2-13}
& 10
& $\varphi_n$ & 0.280 & 0.303 & 0.415 & 0.551 & 0.899 & 0.376 & 0.483 & 0.642 & 0.807 & 0.981 \\
& & $\tilde{\varphi}_n$ & 0.665 & 0.594 & 0.650 & 0.731 & 0.939 & 0.759 & 0.768 & 0.836 & 0.923 & 0.991 \\
\midrule

$(1.8,1,1,1,1)$ & 2
& $\varphi_n$ & 0.244 & 0.448 & 0.739 & 0.969 & 1.000 & 0.472 & 0.760 & 0.957 & 0.996 & 1.000 \\
& & $\tilde{\varphi}_n$ & 0.579 & 0.713 & 0.889 & 0.984 & 1.000 & 0.769 & 0.913 & 0.992 & 0.998 & 1.000 \\
\cmidrule(l){2-13}
& 10
& $\varphi_n$ & 0.343 & 0.517 & 0.772 & 0.981 & 1.000 & 0.497 & 0.764 & 0.948 & 0.997 & 1.000 \\
& & $\tilde{\varphi}_n$ & 0.752 & 0.796 & 0.918 & 0.995 & 1.000 & 0.868 & 0.933 & 0.975 & 0.999 & 1.000 \\
\midrule

$(1.2,1,1,1,1)$ & 2
& $\varphi_n$ & 0.243 & 0.270 & 0.300 & 0.337 & 0.472 & 0.316 & 0.356 & 0.422 & 0.542 & 0.753 \\
& & $\tilde{\varphi}_n$ & 0.625 & 0.599 & 0.615 & 0.657 & 0.775 & 0.677 & 0.724 & 0.768 & 0.831 & 0.932 \\
\cmidrule(l){2-13}
& 10
& $\varphi_n$ & 0.293 & 0.273 & 0.274 & 0.331 & 0.505 & 0.361 & 0.352 & 0.442 & 0.537 & 0.755 \\
& & $\tilde{\varphi}_n$ & 0.716 & 0.656 & 0.633 & 0.653 & 0.790 & 0.807 & 0.742 & 0.783 & 0.850 & 0.936 \\
\bottomrule
\end{tabular}
\end{table}

\clearpage
\section{Additional numerical experiments}\label{sec:A3}

To assess robustness under heteroscedasticity observed in Section~\ref{sec:6}, we conduct additional simulations focusing on type I error rates.

The simulation settings are the same as in Section~\ref{sec:5} for the binary explanatory variable case, except for the covariance matrices.
The covariance matrix $\bm{\Sigma}_i$ of the $i$-th specimen depends on $i \in \{1, \ldots, n\}$.
Specifically, we set $\bm{\Sigma}_i = \diag(\bm{\lambda})$ for $i \in \{1,\ldots,n/2\}$ and $\bm{\Sigma}_i = 2 \diag(\bm{\lambda})$ for $i \in \{n/2 + 1,\ldots,n\}$ with
\begin{align*} 
\bm{\lambda}^\top 
\in \{ (10,8,6,4,2),
(10,1,1,1,1), 
(2,1.5,1.5,1.5,1), 
(1.8,1,1,1,1), 
(1.2,1,1,1,1) \}.
\end{align*}
Under this setting, we compare $\varphi_n$, $\tilde{\varphi}_n$, and $\varphi_{\text{BFN}}$.

Empirical type I error rates are reported in Table~\ref{tab:size3}. 
The results are similar to those under homoscedasticity. 
Specifically, the size of $\varphi_n$ remains more conservative under mild heteroscedasticity than $\tilde{\varphi}_n$ and $\varphi_{\text{BFN}}$.
Similar tendencies are also observed under the same alternative hypotheses as in Section~\ref{sec:5}.

\begin{table}[htbp]
\centering
\small
\setlength{\tabcolsep}{2.8pt}
\renewcommand{\arraystretch}{1.1}
\caption{Empirical type I error rates under heteroscedasticity with eigenvalues $(\lambda_1,\ldots,\lambda_5)$ and $c/n = n^{-1} \| \bm{X}_n \bm{\alpha} \|^2$ for $n \in \{20, 50, 100, 200, 500\}$.
Bold entries indicate rates not exceeding the significance level $\alpha_0 = 0.05$.}
\label{tab:size3}
\begin{tabular}{llc ccccc}
\toprule
$(\lambda_1,\ldots,\lambda_5)$ & $c/n$ & Method & 20 & 50 & 100 & 200 & 500 \\
\midrule

$(10,8,6,4,2)$ & 2
& $\varphi_n$ & 0.051 & 0.066 & 0.085 & 0.091 & 0.051 \\
& & $\tilde{\varphi}_n$ & 0.092 & \textbf{0.045} & 0.052 & 0.052 & \textbf{0.045} \\
& & $\varphi_{\mathrm{BFN}}$ & 0.607 & 0.425 & 0.379 & 0.281 & 0.151 \\
\cline{2-8}
& 10
& $\varphi_n$ & 0.194 & 0.181 & 0.130 & 0.116 & 0.075 \\
& & $\tilde{\varphi}_n$ & 0.307 & 0.201 & 0.135 & 0.114 & 0.095 \\
& & $\varphi_{\mathrm{BFN}}$ & 0.808 & 0.691 & 0.496 & 0.391 & 0.198 \\
\midrule

$(10,1,1,1,1)$ & 2
& $\varphi_n$ & \textbf{0.028} & \textbf{0.028} & \textbf{0.042} & \textbf{0.037} & \textbf{0.043} \\
& & $\tilde{\varphi}_n$ & 0.070 & \textbf{0.047} & \textbf{0.047} & \textbf{0.041} & \textbf{0.046} \\
& & $\varphi_{\mathrm{BFN}}$ & 0.318 & 0.132 & 0.086 & 0.068 & \textbf{0.046} \\
\cline{2-8}
& 10
& $\varphi_n$ & 0.066 & \textbf{0.046} & 0.058 & 0.052 & 0.057 \\
& & $\tilde{\varphi}_n$ & 0.170 & 0.090 & 0.073 & 0.068 & 0.059 \\
& & $\varphi_{\mathrm{BFN}}$ & 0.402 & 0.137 & 0.081 & 0.075 & 0.056 \\
\midrule

$(2,1.5,1.5,1.5,1)$ & 2
& $\varphi_n$ & 0.136 & 0.155 & 0.126 & 0.103 & 0.060 \\
& & $\tilde{\varphi}_n$ & 0.308 & 0.236 & 0.198 & 0.149 & 0.098 \\
& & $\varphi_{\mathrm{BFN}}$ & 0.886 & 0.793 & 0.720 & 0.555 & 0.294 \\
\cline{2-8}
& 10
& $\varphi_n$ & 0.253 & 0.183 & 0.144 & 0.102 & 0.086 \\
& & $\tilde{\varphi}_n$ & 0.570 & 0.366 & 0.267 & 0.177 & 0.123 \\
& & $\varphi_{\mathrm{BFN}}$ & 0.954 & 0.887 & 0.773 & 0.615 & 0.329 \\
\midrule

$(1.8,1,1,1,1)$ & 2
& $\varphi_n$ & 0.112 & 0.064 & 0.071 & \textbf{0.040} & 0.068 \\
& & $\tilde{\varphi}_n$ & 0.283 & 0.160 & 0.130 & 0.083 & 0.075 \\
& & $\varphi_{\mathrm{BFN}}$ & 0.809 & 0.613 & 0.399 & 0.195 & 0.086 \\
\cline{2-8}
& 10
& $\varphi_n$ & 0.159 & 0.116 & 0.074 & 0.057 & 0.069 \\
& & $\tilde{\varphi}_n$ & 0.473 & 0.274 & 0.160 & 0.101 & 0.084 \\
& & $\varphi_{\mathrm{BFN}}$ & 0.916 & 0.743 & 0.508 & 0.258 & 0.106 \\
\midrule

$(1.2,1,1,1,1)$ & 2
& $\varphi_n$ & 0.207 & 0.223 & 0.199 & 0.185 & 0.118 \\
& & $\tilde{\varphi}_n$ & 0.492 & 0.460 & 0.401 & 0.327 & 0.213 \\
& & $\varphi_{\mathrm{BFN}}$ & 0.936 & 0.944 & 0.910 & 0.848 & 0.673 \\
\cline{2-8}
& 10
& $\varphi_n$ & 0.312 & 0.258 & 0.207 & 0.174 & 0.120 \\
& & $\tilde{\varphi}_n$ & 0.665 & 0.525 & 0.464 & 0.348 & 0.224 \\
& & $\varphi_{\mathrm{BFN}}$ & 0.988 & 0.957 & 0.945 & 0.892 & 0.741 \\
\bottomrule
\end{tabular}
\end{table}

\end{document}